\documentclass[10pt,journal,compsoc]{IEEEtran}
\ifCLASSOPTIONcompsoc
  \usepackage[nocompress]{cite}
\else
  \usepackage{cite}
\fi
\ifCLASSINFOpdf
   \usepackage[pdftex]{graphicx}
\else
\fi
\usepackage{amsmath}
\usepackage{amsthm}
\newtheorem{lemma}{Lemma}
\newtheorem{proposition}{Proposition}
\newtheorem{corollary}{Corollary}

\usepackage{algorithmic}

\usepackage{array}

\ifCLASSOPTIONcompsoc
 \usepackage[caption=false,font=footnotesize,labelfont=sf,textfont=sf]{subfig}
\else
 \usepackage[caption=false,font=footnotesize]{subfig}
\fi
\begin{document}
%
\title{To Help or Disturb: Introduction of Crowdsourced WiFi to 5G Networks}
%
%
%
%

\author{Shugang~Hao,
        and~Lingjie~Duan,~\IEEEmembership{Senior Member,~IEEE}
\IEEEcompsocitemizethanks{\IEEEcompsocthanksitem S. Hao and L. Duan are with the Pillar of Engineering Systems and Design, Singapore University of Technology and Design, Singapore, 487372 Singapore.\protect\\
E-mail: shugang\_hao@mymail.sutd.edu.sg, lingjie\_duan@sutd.edu.sg.}
\thanks{Lingjie Duan was
supported by the Ministry of Education, Singapore, under its Academic Research Fund Tier
2 Grant (Project No. MOE-T2EP20121-0001).}}

\IEEEtitleabstractindextext{%
\begin{abstract}
After upgrading to 5G, a network operator still faces congestion when providing the ubiquitous wireless service to the crowd. To meet users' ever-increasing demand, some other operators (e.g., Fon) have been developing another crowdsourced WiFi network to combine many users' home WiFi access points and provide enlarged WiFi coverage to them. While the 5G network experiences negative network externality, the crowdsourced WiFi network helps offload traffic from 5G and its service coverage exhibits positive externality with its subscription number. To our best knowledge, we are the first to investigate how these two heterogeneous networks of diverse network externalities co-exist from an economic perspective. We propose a dynamic game theoretic model to analyze the hybrid interaction among the 5G operator, the crowdsourced WiFi operator, and users. Our user choice model with WiFi's complementarity for 5G allows users to choose both services, departing from the traditional economics literature where a user chooses one over another alternative. Despite of non-convexity of the operators’ pricing problems, we prove that the 5G operator facing severe congestion may purposely lower his price to encourage users to add-on WiFi to offload, and he benefits from the introduction of crowdsourced WiFi. However, 5G operator with mild congestion tends to charge users more and all the users' payoffs may decrease.
\end{abstract}

}

\maketitle

\IEEEdisplaynontitleabstractindextext

%
\IEEEpeerreviewmaketitle

\IEEEraisesectionheading{\section{Introduction}\label{section: TMC-1}}

With rapid proliferation of mobile devices and wireless technology, mobile data traffic has been ever-increasing at an explosive speed. Though major network operators have upgraded their cellular networks to 5G, the network capacity still grows at a pace much slower than the data traffic demand {(e.g., \cite{cc, fk})}.  As compared to 5G, WiFi is a mature short-range wireless technology to control congestion and is easy to enable with relatively low deployment cost. Yet, an individual WiFi access point (AP) has small service coverage (even supported by the latest gigabit WiFi amendments in 802.11ac/ad/ax), and it is difficult to deploy a ubiquitous WiFi network due to the formidably high cost to cover every corner\cite{kosek2017coexistence}.

Recently, the crowdsourced WiFi community network has emerged as a promising approach to address the small coverage problem for WiFi. Non-traditional service providers such as Fon and OpenSpark have been developing another WiFi social community network to combine many users' home WiFi access points and provide enlarged WiFi coverage to these users.  In this crowdsourced WiFi network, users following equal-reciprocal principle are motivated to share their home APs with each other and enjoy accessing to the others' APs when roaming around. For example, Fon's WiFi network has included over 23 million APs and is fast expanding in the world to cover many populous and crowded places, where the 5G network faces congestion in such places to serve the crowd especially in rush hours. Note that though greatly enlarged by crowdsourcing separate home APs, the crowdsourced WiFi community network's coverage is still not comparable to the ubiquitous 5G coverage, and is thus viewed as a complement as backhaul to the 5G service instead of a substitutable alternative.

The story of co-existence between 5G and crowdsourced WiFi networks is happening in many countries and places, e.g., between British Telecom and Fon in the United Kingdom, and between Telenor and OpenSpark in Finland (e.g., \cite{ballon2009business, iosifidis2017efficient}). When the negative 5G network externality meets the crowdsourced WiFi's positive externality, we wonder \textit{whether the 5G network operator gains from the introduction of the crowdsourced WiFi network by another self-interested company, and how the 5G operator should adapt his service operation especially his service pricing strategy.} We also wonder the impact on the users' payoffs given another WiFi service option with complementarity for 5G.

To answer such questions from an economic perspective, there are two main technical challenges to address. 
\begin{itemize}
    \item \emph{User choice model with WiFi's complementarity for 5G:} We first need to model and analyze a new user choice model with complementarity given the service options of 5G and 5G+WiFi. This model departs from the traditional economics literature for competing products or services \cite{mas1995microeconomic}. There people think of differentiated substitutes and most users choose one over another from a set of alternatives. Our unique model is somewhat related to the recent literature on complementarity or vertical differentiation (e.g., \cite{armstrong2007recent, li2021impact}). Yet the choice models there are too generic to fit any wireless or network feature, and the service qualities there are fixed without considering any interplay between the incumbent and add-on services' users. Differently, in our problem a user to add on WiFi service helps reduce 5G service congestion and improve the WiFi service coverage.   
    
    \item \emph{When negative meets positive network externality for interactive service operation:} After  considering the users' choice model with complementarity, the 5G and crowdsourced network operators need to decide their economic decisions to maximize their own profits, respectively. Such interactive decision-making becomes involved to analyze when negative meets positive network externality. For example, the 5G operator may purposely lower his service price to motivate more users to add on WiFi service for offloading cellular traffic and reducing his network congestion. This should also be taken into account for the crowdsourced WiFi operator's strategic decision-making. There are few studies in the literature to consider both the positive and negative network externalities (e.g., \cite{gong2017social, xiong2018competition, wang2018mobile}). Yet \cite{gong2017social, xiong2018competition, wang2018mobile} only focused on a single network's central decision-making to balance these two effects.
    
\end{itemize}

 Similarly, existing works on WiFi offloading  largely assume that a cellular network operator deploys his own WiFi network to offload data traffic (e.g., \cite{li2018mobile, iosifidis2014double, lee2014economics}) and focuses on central operation. Yet in practice, we observe the emerging trend of crowdsourced WiFi networks (e.g., Fon) and their operations such as pricing and deployment are not controlled by the 5G operator. Further, unlike traditional WiFi, this crowdsourced WiFi network exhibits positive network externality and its network coverage increases as more users join and share their APs. Note that \cite{gao2014bargaining} only studied how a 5G operator bargains with a traditional WiFi network to decide the offloaded traffic amount and associated payment, by ignoring the crowdsouced WiFi's positive network externality and the users' choices on both services. \cite{manshaei2008wireless, li2020optimal} ignored the negative network externality of the 5G network due to congestion.

Our key novelty and main contributions are summarized as follows:
\begin{itemize}
    \item \emph{To help or disturb: introduction of crowdsourced WiFi to 5G networks:} To our best knowledge, we are the first to investigate how the 5G and crowdsourced networks of diverse network externalities co-exist from an economic perspective. While the incumbent 5G network experiences negative network externality, the crowdsourced WiFi network helps offload heavy traffic from 5G and exhibits positive externality: its service coverage improves as more users join and contribute. We present new analytical studies to tell whether the 5G network operator and users benefit from the introduction of the crowdsourced WiFi network by another self-interested company, and how the 5G operator should best adapt his pricing in service operation.

    \item \emph{Dynamic game-theoretical modelling and analysis:} we propose a dynamic game theoretic model to analyze the hybrid interaction among the 5G operator, the (crowdsourced) WiFi operator, and users: in the first stage the 5G and WiFi operators respectively  decide their service prices for maximizing their own profits according to the network capacity and deployment cost, and in the second stage users of different congestion sensitivities decide to subscribe to 5G, 5G+WiFi or neither. Note that our user choice model with complementarity  departs from the traditional economics literature \cite{mas1995microeconomic}, where a user chooses one over another from a set of alternatives. Here a user’s choice in a service affects all the other users in both services, and it becomes more involved to analyze in the backward induction. 
    
        \item \emph{Impacts on 5G service operation and user payoff:} The
analysis of the dynamic game theoretic model is involved due to solving high-degree polynomial equations for users' choice partition and the non-convexity of the operators' pricing problems. Despite of this, we successfully prove at the equilibrium that the 5G operator gains more profit after the introduction of the crowdsourced WiFi network. If the 5G operator has non-small capacity, we prove that he will charge users more in 5G service after the introduction of the crowdsourced WiFi. Otherwise, he will purposely charge users less to motivate more users to add-on WiFi and help reduce 5G congestion. 
 Perhaps surprisingly, we prove that the introduction of crowdsourced WiFi can reduce all the users' payoffs, which happens when the 5G operator has non-small network capacity. Finally, we further consider the congestion effect in the WiFi network and show the main results still hold. 
   
    \end{itemize} 
    
The rest of this paper is organized as follows. In Section~\ref{section:TMC-2}, we introduce system models and dynamic game formulation before and after the introduction of the crowdsourced WiFi. In Section~\ref{section:TMC-3}, we analyze equilibrium of the dynamic game before the introduction of the crowdsourced WiFi. In Sections~\ref{section:TMC-4} and \ref{section:TMC-5}, we analyze the user subscription equilibrium in Stage II and then the 5G and WiFi operators' pricing after the introduction of the crowdsourced WiFi. Section~\ref{section:TMC-6} presents numerical results to further analyze the impacts of the crowdsouced WiFi. Section~\ref{section:TMC-7} extends our model and analysis to the case with congestion effect in the crowdsouced WiFi. Section~\ref{section:TMC-8} finally concludes the paper.

\begin{figure}[!t]
\centering
\includegraphics[scale = 0.38]{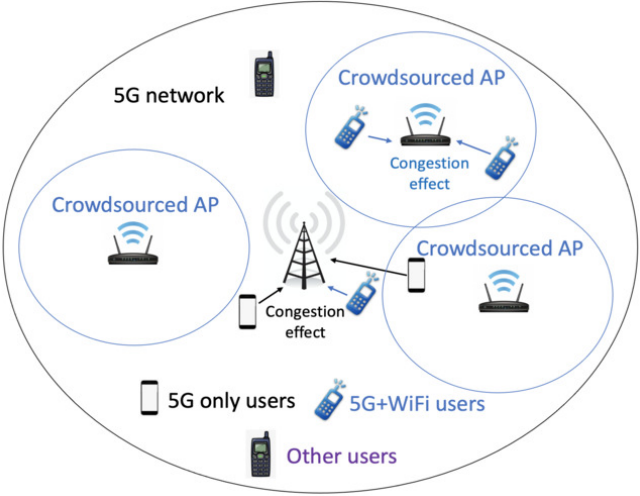}
\caption{System model for the co-existence of the 5G and crowdsourced WiFi networks, where users of different congestion sensitivities decide to subscribe to 5G, 5G+WiFi, or neither. The crowdsourced WiFi network's coverage enlarges as more users choose to add on this service to their 5G contracts. }
\label{fig}
\end{figure}
\section{System Model and Problem Formulation}\label{section:TMC-2}

When the 5G co-exists with the crowdsourced WiFi network, Fig. 1 shows that users are generally partitioned into three groups: 5G only users, 5G+WiFi users, and others with reservation payoff $\bar{u}\geq 0$. The reservation payoff is obtained by choosing nothing or 
 under another benchmark network not comparable to the WiFi service. In this illustrative example, the three users marked in blue (e.g., with high sensitivity to 5G network congestion) choose to add on the WiFi option and share their home APs to each other in the Fon's crowdsourced network.  Each subscriber to either 5G or 5G+WiFi service should have greater payoff than $\bar{u}$. Motivated by the practice, a user will not choose WiFi service only as its coverage (even under all APs' contributions) is small as compared to the ubiquitous coverage by 5G\footnote{City of London WIFI Access Point Locations, https://www.cityoflondon.gov.uk/assets/Business/wifi-access-point-locations-2019.pdf.}. 

In this section, for ease of presentation and comparison,  we respectively present the system models before and after the introduction of  the crowdsourced WiFi to the wireless market. The former case is just used as a benchmark to compare with the latter.

\subsection{Benchmark Case Before The Introduction Of The Crowdsourced WiFi }\label{section:TMC-2.1}
We consider $N$ users in total as potential subscribers to any wireless service. Before the introduction of the crowdsourced WiFi, only the 5G network operator provides wireless service to users at price $\bar{p}_1$ (e.g. monthly), and he has limited network capacity $Q$ to serve the users. We denote the users' proportion or fraction of 5G subscription as $\bar{x}_1\in[0,1]$ and the total subscriber number is $N\bar{x}_1$, depending on price $\bar{p}_1$ and the user choice model as detailed later. Thus, we can rewrite the subscription fraction as a function of price, i.e., $\bar{x}_1=\bar{x}_1(\bar{p}_1)$, and the 5G operator's profit is:   
\begin{equation}
\bar{\pi}_1(\bar{p}_1) = N\bar{x}_1(\bar{p}_1) \bar{p}_1, \label{e2}
\end{equation}
which is a product of the 
number of subscribers and the price per subscriber.

We model the practice that users have different sensitivities to the 5G congestion, and the 5G operator only knows that  each user's sensitivity degree $\theta$, as in the related literature (e.g., \cite{gibbens2000internet, tran2013optimal, ren2012entry}). Generally, $\theta$ in the normalized range [0,1] follows continuous cumulative distribution function (CDF) $F(\cdot)$, where $F(0) = 0$ and $F(1) = 1$. {A user with higher $\theta$ travels more in the populous areas (e.g., shopping malls in downtown) to incur greater congestion cost, and is more reluctant to join.} After paying price $\bar{p}_1$ to join the network, the user obtains the normalized positive value $V_1$ (e.g., for mobile Internet access) and experiences a congestion cost. Following \cite{gibbens2000internet}, we model the user's congestion cost as $\frac{N\bar{x}_1}{Q}\theta$, which increases with his sensitivity $\theta$ and the total subscriber number $N\bar{x}_1$ and decreases with the limited 5G capacity $Q$ to share {(due to collisions and conflicts between signals in using the limited capacity)}. Thus, the user's payoff of 5G subscription is given by:   
\begin{equation}
    \bar{u}_1(\theta) = V_1 - \frac{N\bar{x}_1}{Q}\theta - \bar{p}_1. \label{e1}
\end{equation}
From \eqref{e1}, we can tell that only users with not large congestion sensitivity $\theta$ have $\bar{u}_1(\theta) \geq \bar{u}$ and choose to join the 5G network. 

{ In practice, the network operators have more power to lead as compared to the users as followers \cite{mas1995microeconomic}. Following the classic economic model of service pricing (e.g., \cite{gibbens2000internet, manshaei2008wireless, tran2013optimal, ren2012entry})}, we then formulate a Stackelberg game between the 5G operator and the $N$ users as follows:
\begin{itemize}
    \item Stage I: The 5G network operator strategically  decides and announces subscription fee $\bar{p}_1$ to $N$ potential users.
    \item Stage II: After observing the price, each user  with personalized  congestion sensitivity $\theta$ chooses to join the 5G network or not, depending on  whether his payoff in \eqref{e1} is no less than the fixed reservation payoff (i.e.,  $\bar{u}_1(\theta) \geq \bar{u}$) or not.  
\end{itemize}

We will analyze this benchmark case using backward induction  in Section~\ref{section:TMC-3}, and the equilibrium results will be compared to those based on the next subsection's model.

\subsection{System Model After The Introduction Of The Crowdsourced WiFi}\label{section:TMC-2.2}
As shown in Fig. 1, with existence of the crowdsourced WiFi, users with access to 5G can further choose to add on the crowdsourced WiFi service or not. 

We denote new $x_1 \in [0, 1]$ as the fraction of $N$ users joining the 5G network only,  and $x_2 \in [0, 1]$ as the fraction of joining both 5G+WiFi. For an individual user choosing to join the crowdsourced WiFi network, he contributes a normalized positive addition $\alpha \in (0,1)$ to the network coverage using his home AP and thus the overall coverage of the crowdsourced WiFi network is $\alpha x_2$ to tell the positive externality,  which is also validated in \cite{wang2018mobile} and \cite{manshaei2008wireless}. In the current service practice of Fon company, any user can easily check the latest network coverage to access\footnote{https://fon.com/maps/.}.

Different from the benchmark case before the introduction of the crowdsourced WiFi, now the 5G congestion is also affected by the fraction $x_2$ in the crowdsourced WiFi network to possibly offload the 5G traffic. Each 5G+WiFi user finds himself in the WiFi coverage with probability $\alpha x_2$ to use the congestion-free WiFi, and will only choose the 5G service outside with the rest probability $1-\alpha x_2$. Unlike \eqref{e1}, now only $x_1 + x_2 (1 - \alpha x_2)$ fraction of users will share the limited capacity $Q$ and contribute to the negative congestion effect in the 5G network. Next, we introduce the user choice model with complementarity.  

On one hand, by choosing 5G only, a user with congestion sensitivity $\theta$ pays price $p_1$ to obtain the mobile access benefit $V_1$. No matter inside or outside the WiFi network coverage, he faces a congestion cost   $\frac{N(x_1 + x_2 (1 - \alpha x_2)  )}{Q}\theta$ all the time. This cost term is increasing in $x_1$ but may not be monotonic in $x_2$, as an increase of $x_2$ enlarges the WiFi network coverage but adds more 5G traffic outside the WiFi coverage. Then we can write down the user's payoff of 5G subscription only as 
\begin{equation}
    u_1(\theta) = V_1 - \frac{N(x_1+x_2(1 - \alpha x_2))}{Q}\theta - p_1. \label{e3}
\end{equation}

On the other hand, { given the add-on option of WiFi to his 5G service, the user can churn to WiFi service when facing severe 5G congestion.} He  {has benefit $V_1$ outside the WiFi coverage with probability $1-\alpha x_2$, and gains $V_2$ inside the WiFi coverage with probability $\alpha x_2$, where $V_1 \geq V_2$ implies 5G provides more benefits.} Besides, he needs to pay extra $p_2$ to the WiFi operator, but does not face 5G congestion in the WiFi coverage with probability $\alpha x_2$.  Thus, his payoff becomes: 
\begin{align}
    u_2(\theta) = \ &(1-\alpha x_2)V_1 + \alpha x_2 V_2 - p_1 - p_2 \nonumber \\
                  & - (1 -\alpha x_2)\frac{N(x_1+x_2(1 - \alpha x_2))}{Q}\theta, \label{e4}
\end{align}

Here we do not consider the congestion in the WiFi network, since the ISM band is enough to support local users with tolerable congestion and many APs remain under-utilized most of the time \cite{duan2014pricing}. 
We will still extend our model to consider WiFi congestion in Section~\ref{section:TMC-7}, where our main results hold as long as the WiFi congestion is not as large as the 5G congestion. 

{ The user compares his payoffs (in \eqref{e3}, \eqref{e4} and $\bar{u}$) under the three service options (5G only, 5G+WiFi, and neither), and optimally chooses the one that yields the maximum payoff. This is similar to the literature of discrete service choices (e.g., \cite{gibbens2000internet, manshaei2008wireless, tran2013optimal, ren2012entry}).} Unlike \eqref{e1}, now it is less straightforward to  compare \eqref{e3} and \eqref{e4} for users' choices under different $\theta$ realizations. One can imagine that those with trivial congestion sensitivity do not want to add on the WiFi option for saving expense, and those sensitive users may want to add on to enjoy a better service quality. Mathematically, the user will choose 5G only if $u_1(\theta) \geq u_2(\theta)$ and $u_1(\theta) \geq \bar{u}$, and 5G+WiFi if $u_2(\theta) > u_1(\theta)$ and $u_2(\theta) \geq \bar{u}$.

Both $x_1$ and $x_2$ fractions of users pay the 5G operator with price $p_1$, and thus the 5G profit changes from \eqref{e2} to:
\begin{equation}
    \pi_1 = N (x_1+x_2) p_1, \label{e5}
\end{equation}
where not only $x_1$ and $x_2$ are affected by $p_1$ and $p_2$ but they also affect each other. For example, an increase of $x_2$ in the crowdsourced WiFi network helps offload more 5G traffic and may motivate a larger $x_1+x_2$ to stay in the 5G network.

For the crowdsourced WiFi operator, he selfishly decides price $p_2$ and only collects payments from the fraction $x_2$ of users, and his profit is: 
\begin{equation}
    \pi_2 = N x_2 (p_2 - c), \label{e6}
\end{equation}
where we also model the WiFi deployment cost $c$ per user/AP to install and add to the crowdsourced network as in \cite{li2020optimal}.

 {Following the classic economic model of service pricing (e.g., \cite{gibbens2000internet, manshaei2008wireless, tran2013optimal, ren2012entry})}, 
we are now ready to formulate a dynamic game as follows:
\begin{itemize}
    \item Stage I: The 5G and the crowdsourced WiFi network operators simultaneously decide and announce subscription fee $p_1 \geq 0$ and $p_2 \geq c$, respectively, to $N$ potential users.
    \item Stage II: After observing the two prices, each user with personalized congestion sensitivity $\theta$ chooses to join the 5G network only if his payoff in \eqref{e3} is no less than the fixed reservation payoff and no less than that in \eqref{e4},  (i.e.,  $u_1(\theta) \geq \bar{u}$ and $u_1(\theta) \geq u_2(\theta)$), or joins 5G+WiFi if his payoff in \eqref{e4} is no less than the fixed reservation payoff and greater than that in \eqref{e3}, (i.e.,  $u_2(\theta) \geq \bar{u}$ and $u_2(\theta) > u_1(\theta)$), or joins neither with the fixed reservation payoff $\bar{u}$.
\end{itemize}

{ The dynamic game is under complete information. The operators' service prices, the network capacity and related technological specifications are specified in the service contracts with users to expect the service benefit and congestion. Note that for crowdsourced WiFi service, the major operator Fon announces its roadmap and WiFi coverage online for its users to check and access. Operators also provide flexible contracts or trial sessions for users to explore and learn before long-term commitment. }

In Sections~\ref{section:TMC-4} and \ref{section:TMC-5}, we will apply backward induction to analyze this involved dynamic game due to the complex user choice model when positive meets negative network externality. Though it is difficult to obtain the closed-form solutions, we can still prove structural results there to tell if the 5G operator and users benefit from the introduction of the crowdsourced WiFi at the equilibrium. 

\begin{figure*}
\centering
\subfloat[5G price equilibrium $\bar{p}_1^*$ in \eqref{p_5} versus the 5G network capacity $Q$]{\includegraphics[scale = 0.25]{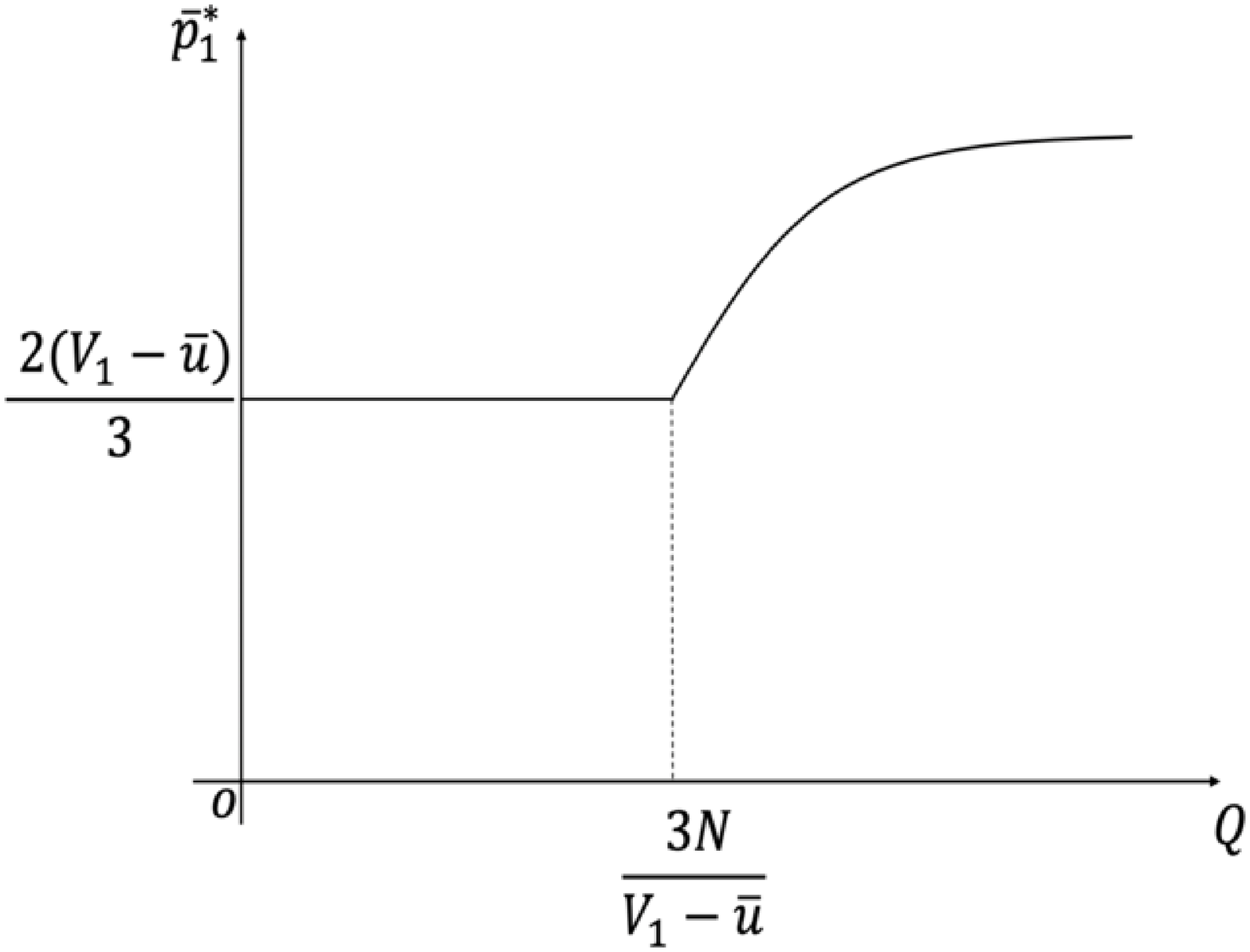}%
\label{fig_first}}
\hfil
\subfloat[Users' subscription $\bar{x}_1^*$ in \eqref{x_5} versus capacity $Q$]{\includegraphics[scale = 0.25]{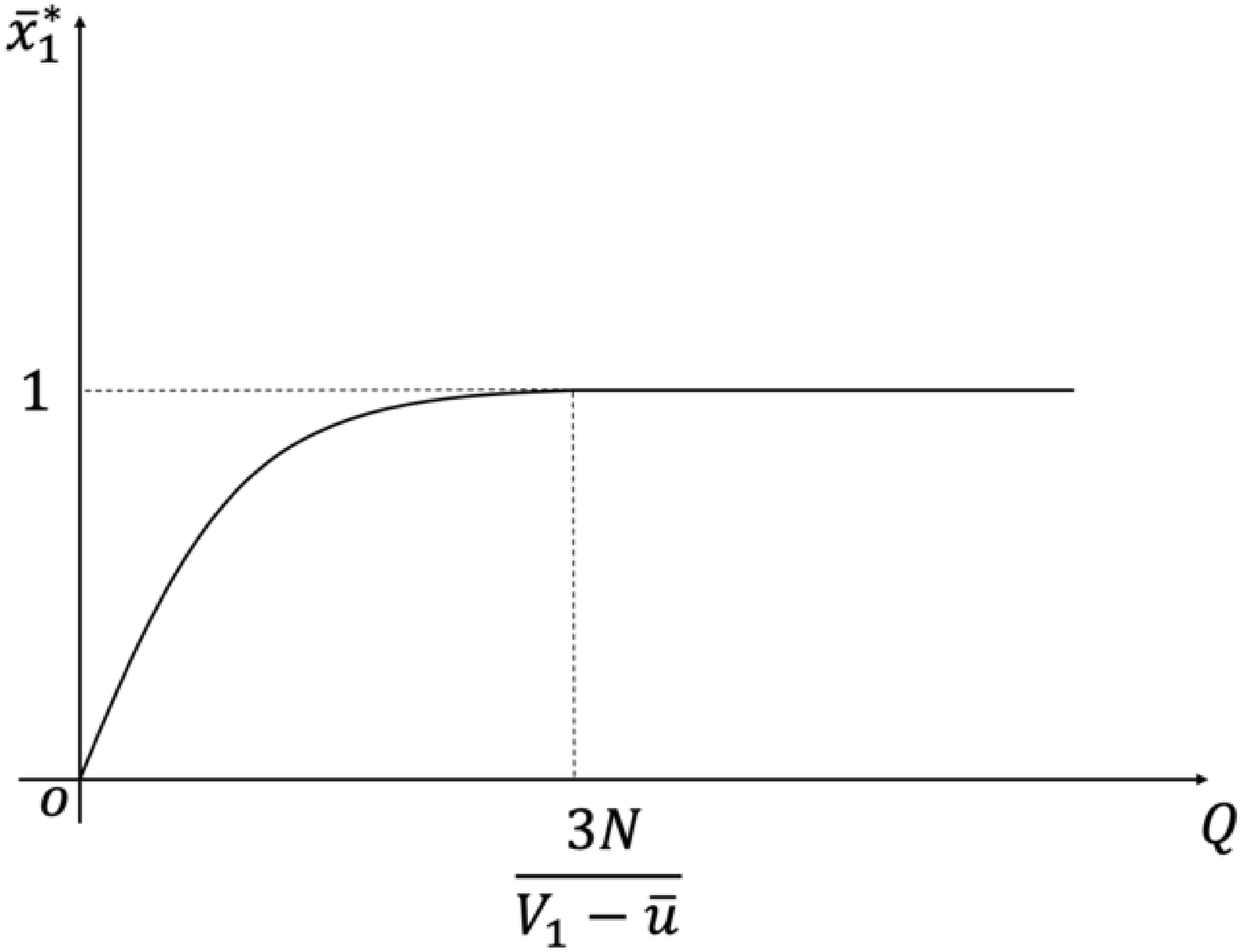}%
\label{fig_second}}
\caption{Equilibrium price and subscription fraction  versus the 5G network capacity $Q$ before the introduction of the crowdsourced WiFi. Here we plot given users' congestion sensitivity $\theta$ following uniform distribution.}
\label{fig333}
\end{figure*}

\section{Equilibrium Analysis Of The Benchmark Case}\label{section:TMC-3}
The Stackelberg game model in Section~\ref{section:TMC-2.1} can be analyzed by backward induction from Stage II, as there is only 5G network with negative externality.  Denote $F^{-1}(\cdot)$ as the inverse function of CDF $F(\cdot)$. Since only users of low 5G congestion sensitivity will join the 5G network,  we expect  $\bar{u}_1(\theta = F^{-1}(\bar{x}_1^*))  = \bar{u}$ in \eqref{e1} for solving equilibrium fraction $\bar{x}_1^*$ as a function of 5G price $\bar{p}_1$ in \eqref{e2}. 

\begin{lemma}\label{lemma:TMC-1}
Given any 5G price $\bar{p}_1$ in Stage I, only users with low congestion sensitivity $\theta\leq F^{-1}(\bar{x}_1^*)$ decide to join the 5G network, where 
\begin{equation}
  \bar{x}_1^* =  
  \begin{cases}
   1, &\text{if} \; \bar{p}_1 \leq V_1 - \bar{u} - N/Q,  \\
   F(\bar{\theta}^*), &\text{if} \; V_1 -\bar{u} - N/Q < \bar{p}_1 < V_1 - \bar{u},  \\
   0, &\text{if} \; \bar{p}_1 \geq V_1 - \bar{u}, \label{e333}
\end{cases}
\end{equation}
and cutoff $\bar{\theta}^* \in (0, 1)$ is the unique solution to
\begin{equation*}
    V_1 - \frac{N}{Q} F(\bar{\theta}^*)\bar{\theta}^* -  \bar{p}_1 - \bar{u} = 0. \label{t_star}
\end{equation*}
The final 5G subscription fraction $\bar{x}_1^*$ in \eqref{e333} increases with the 5G network capacity $Q$ and decreases with price $\bar{p}_1$.
\end{lemma}

Now we turn to Stage I to find the best price $\bar{p}_1$ for maximizing the 5G operator's profit in \eqref{e2}. By substituting  $\bar{x}_1^*$ in \eqref{e333} into $\bar{\pi}_1$ in \eqref{e2}, the profit objective becomes: 
\begin{equation}
   \bar{\pi}_1 = 
  \begin{cases}
   N \bar{p}_1, \;\; &\text{if} \; \bar{p}_1 \leq V_1 - \bar{u} - N/Q,  \\
   N \bar{p}_1 F(\bar{\theta}^*), \;\; &\text{if} \; V_1 -\bar{u} - N/Q < \bar{p}_1 < V_1 - \bar{u},  \\
   0, \;\; &\text{if} \; \bar{p}_1 \geq V_1 - \bar{u}. \label{e444}
\end{cases}
\end{equation}

One may rewrite the price $\bar{p}_1$ as a function of cutoff $\bar{\theta}^*$, and simplify the profit-maximization problem \eqref{e444}. Yet the problem is still non-convex and is difficult to solve analytically given general $F(\theta)$ distribution. Under certain assumption on the users' congestion sensitivity distribution, we solve the problem analytically. If the assumption does not hold, we can still solve it numerically. 

\begin{proposition}\label{prop:TMC-1}
Assume PDF $f(\theta)$ of users' congestion sensitivity is continuous and non-increasing, which applies to uniform, exponential and Pareto distributions. At the equilibrium of the Stackelberg game, the 5G price equilibrium in Stage I is:
\begin{equation*}
  \bar{p}_1^* = 
  \begin{cases}
   V_1 - \bar{u} - \frac{N}{Q}\bar{\theta}^*F(\bar{\theta}^*),  \;\; &\text{if} \; Q < \bigg(2+\frac{1}{f(1)}\bigg)\frac{N}{V_1-\bar{u}},   \\
   V_1 - \bar{u} - \frac{N}{Q}, \;\; &\text{if} \; Q \geq \bigg(2+\frac{1}{f(1)}\bigg)\frac{N}{V_1-\bar{u}}, 
\end{cases}
\end{equation*}
and the fraction of 5G subscribers $\bar{x}_1^*$ is:
\begin{equation*}
  \bar{x}_1^*  = 
  \begin{cases}
    F(\bar{\theta}^*), \;\; &\text{if} \;Q < \bigg(2+\frac{1}{f(1)}\bigg)\frac{N}{V_1-\bar{u}},   \\
    1, \;\; &\text{if} \; Q \geq \bigg(2+\frac{1}{f(1)}\bigg)\frac{N}{V_1-\bar{u}},  
\end{cases}
\end{equation*}
where $\bar{\theta}^* \in (0, 1)$ is the solution to 
\begin{equation*}
    f(\bar{\theta}^*)\bigg(V_1 - \bar{u} - \frac{2N}{Q}F(\bar{\theta}^*)\bar{\theta}^*\bigg) - \frac{N}{Q}F^2(\bar{\theta}^*) = 0.
\end{equation*}

Specifically, if $\theta$ is uniformly distributed in $[0, 1]$, the 5G price in Stage I reduces to:
\begin{equation}
  \bar{p}_1^* = 
  \begin{cases}
   \frac{2}{3}(V_1-\bar{u}),  \;\; &\text{if} \; Q < \frac{3N}{V_1-\bar{u}},   \\
   V_1 - \bar{u} - \frac{N}{Q}, \;\; &\text{if} \; Q \geq \frac{3N}{V_1-\bar{u}},  \label{p_5}
\end{cases} 
\end{equation}
which increases with network capacity $Q$ (as illustrated in Fig. \ref{fig_first}).
The fraction of 5G subscribers $\bar{x}_1^*$ is:
\begin{equation}
  \bar{x}_1^*  = 
  \begin{cases}
    \sqrt{\frac{V_1-\bar{u}}{3N/Q}}, \;\; &\text{if} \; Q < \frac{3N}{V_1-\bar{u}},   \\
    1, \;\; &\text{if} \; Q \geq \frac{3N}{V_1-\bar{u}},  \label{x_5}
\end{cases}
\end{equation}
which increases in network capacity $Q$ (see Fig. \ref{fig_second}). The 5G network operator's profit $\bar{\pi}_1^*$ is
\begin{equation}
  \bar{\pi}_1^*  = 
  \begin{cases}
    \frac{2}{3}(V_1-\bar{u}) \sqrt{\frac{(V_1-\bar{u})QN}{3}}, \;\; &\text{if} \; Q < \frac{3N}{V_1-\bar{u}},   \\
    N\left(V_1 -\bar{u}- \frac{N}{Q}\right), \;\; &\text{if} \; Q \geq \frac{3N}{V_1-\bar{u}},  
\end{cases} \label{pi_5}
\end{equation}
which increases in network capacity $Q$.
\end{proposition}

As illustrated in Fig. \ref{fig333}, if the capacity is small (i.e., $Q < 3N/(V_1-\bar{u})$), the 5G network is severely congested and the operator can only charge a low price $\bar{p}_1^*=2(V_1-\bar{u})/3$ from users. Having non-small $Q$, its price and profit increase with $Q$.

\section{Stage II Analysis of Dynamic Game After the introduction of the crowdsourced WiFi}\label{section:TMC-4}
For our focused model in Section~\ref{section:TMC-2.2}, we begin with equilibrium analysis of user choices in Stage II, given the prices $p_1$ and $p_2$ decided in Stage I.  Denote the service choice of a user with congestion sensitivity $\theta \in [0, 1]$ as $\phi(\theta) \in \{0, 1, 2\}$, where 
\begin{align*}
    \phi(\theta) = \begin{cases}
    1, &\text{if the user subscribes to 5G only}, \\
    2, &\text{if the user subscribes to 5G+WiFi}, \\ 
    0, &\text{if the user subscribes to neither.}
    \end{cases}
\end{align*}


The analysis of decision $\phi(\theta)$ is involved, as a user's choice in a service affects all the other users in both services. For example, a user to add on WiFi helps reduce the 5G users' congestion and improve the WiFi network coverage. Thus, a user's decision needs to proactively consider any other user's decision in any service.  

To simplify the user choice analysis, we first suppose the equilibrium fractions of users in 5G and 5G+WiFi (i.e., $x_1^*$ and $x_2^*$) are known and stable, and derive each user's choice according to his personalized cost congestion $\theta$, by comparing \eqref{e3} and \eqref{e4}.

\begin{table*}
\begin{equation}
    \hat{p}_2 = \alpha\hat{x}_2\frac{N}{Q}(1-\alpha\hat{x}_2^2)(1-\hat{x}_2), \;\;
    \hat{p}_2' = \alpha \sqrt{\frac{1-\sqrt{\frac{V_1-\bar{u}-\hat{p}_1}{N/Q}}}{\alpha}}\sqrt{\frac{V_1-\bar{u}-\hat{p}_1}{N/Q}}\frac{N}{Q} 
    \bigg(1-\sqrt{\frac{1-\sqrt{\frac{V_1-\bar{u}-\hat{p}_1}{N/Q}}}{\alpha}}\bigg). \label{pro4_}
\end{equation}
\end{table*}
\begin{table*}
\begin{equation}
    \frac{1-\sqrt{\frac{V_1-\bar{u}-\hat{p}_1}{N/Q}}}{\alpha} + \frac{1 - \sqrt{1 - 3\alpha\sqrt{\frac{V_1-\bar{u}-\hat{p}_1}{N/Q}} }}{3\alpha}
    \bigg(\sqrt{\frac{V_1-\bar{u}-\hat{p}_1}{N/Q}} -
     \frac{1+3\alpha\sqrt{\frac{V_1-\bar{u}-\hat{p}_1}{N/Q}}-\sqrt{1-3\alpha\sqrt{\frac{V_1-\bar{u}-\hat{p}_1}{N/Q}}}}{9\alpha}\bigg) - \sqrt{\frac{1-\sqrt{\frac{V_1-\bar{u}-\hat{p}_1}{N/Q}}}{\alpha}} = 0. \label{pro4}
\end{equation}
\hrulefill
\end{table*}

\begin{figure}[!t]
\centerline{\includegraphics[scale = 0.45]{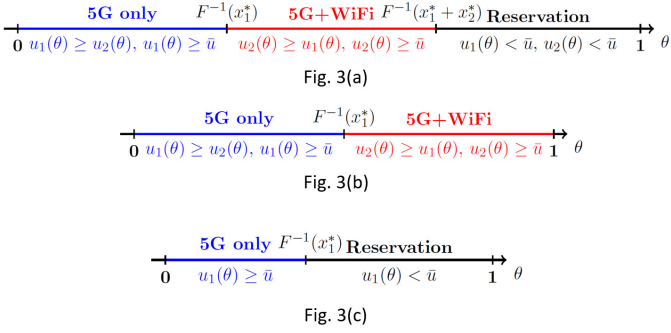}}
\caption{Users' subscription in the heterogeneous networks of diverse network externalities in Stage II of the  dynamic game. Note that it is possible to have $x_1^*+x_2^*<1$ in Fig. 3(a), or  $x_1^*+x_2^*=1$ where all users either choose 5G or 5G+WiFi service in Fig. 3(b).}
\label{fig4}

\end{figure}

\begin{lemma}\label{lemma:TMC-2}
At the equilibrium of Stage II, a user decides his service choice by comparing his personalized congestion cost $\theta\in[0,1]$ to two cutoff points $F^{-1}(x_1^*)$ and $F^{-1}(x_1^*+x_2^*)$ as in the following two cases.  
\begin{itemize}
    \item If there are a positive number of subscribers to choose 5G+WiFi (i.e.,  $x_2^* > 0$), a user's equilibrium choice is (see Fig. \ref{fig4}(a) or Fig. \ref{fig4}(b)):  
   \begin{equation}
   \phi(\theta)\! = \!\!\left\{ 
  \begin{aligned}\label{e11}
   &1, \; \text{if} \;\theta \leq F^{-1}(x_1^*) \;\text{and}\; u_1(\theta) \geq \bar{u},   \\
   &2, \; \text{if} \; \theta \in (F^{-1}(x_1^*), F^{-1}(x_1^* + x_2^*)] \\
        &\;\;\;\;\;\;\; \text{and}\; u_2(\theta) \geq \bar{u}, \\ 
   &0, \; \text{if} \;\theta\! \in\! (F^{-1}(x_1^* + x_2^*), 1],\! u_1(\theta)\! < \! \bar{u}, \\
    &\;\;\;\;\;\;\; \text{and} \;\! u_2(\theta) \!< \!\bar{u}.  \\
\end{aligned}\right. 
\end{equation}
\item If there is no subscriber to 5G+WiFi (i.e.,  $x_2^* = 0$), a user's equilibrium choice is (see Fig. \ref{fig4}(c)):   
   \begin{equation}
   \phi(\theta) = \left\{ 
  \begin{aligned}
   &1, \; \text{if} \;\theta \leq F^{-1}(x_1^*) \;\text{and}\; u_1(\theta) \geq \bar{u},   \\
   &0, \; \text{if} \; \theta \in (F^{-1}(x_1^*), 1], u_1(\theta) < \bar{u}, \\
   &\;\;\;\;\;\;\;\text{and}\; u_2(\theta) < \bar{u}. 
\end{aligned}\right.  \label{e12}
\end{equation}
\end{itemize}
\end{lemma}  

\begin{proof} 
Given $p_1$ and $p_2$ in Stage I, a user of congestion sensitivity $\theta$ compares $u_1(\theta)$ in \eqref{e3} and $u_2(\theta)$ in \eqref{e4} to decide service.
Define their difference as a function of $\theta$: 
\begin{align*}
    f(\theta) &:= u_1(\theta) - u_2(\theta)  \\
              &= p_2 - \alpha x_2^* \bigg( \frac{N(x_1^*+x_2^*(1 - \alpha x_2^*))}{Q}\theta + V_2 - V_1 \bigg).
\end{align*}
For the user with $\theta = 0$, $f(0) = p_2 + \alpha x_2^* (V_1 - V_2) \geq 0$, which means $u_1(0) \geq u_2(0)$ and the user with $\theta = 0$ prefers to join the 5G network than 5G+WiFi. 

Since $f(\theta)$ above linearly decreases with $\theta$ and $f(0) \geq 0$, we have either $f(\theta) \geq 0$ for any $\theta \in [0, 1]$, or $f(\theta) \geq 0$ for $\theta \in [0, F^{-1}(x_1^*)]$ and $f(\theta) < 0$ for $\theta \in (F^{-1}(x_1^*), F^{-1}(x_1^* + x_2^*)]$, where we have $x_1^* \in (0, 1]$, $x_2^* \in (0, 1)$ and $x_1^* + x_2^* \in (0, 1]$. The former case tells that users only consider 5G only and join if $\theta\leq F^{-1}(x_1^*)$ and $u_1(\theta)\geq \bar{u}$, which results in \eqref{e12} with $x_2^*=0$ in Lemma~\ref{lemma:TMC-2}. The latter case results in \eqref{e11} with $x_2^*>0$. 
\end{proof}

Given the users' choices in \eqref{e11} and \eqref{e12} according to cutoff points $F^{-1}(x_1^*)$ and $F^{-1}(x_1^*+x_2^*)$, we are ready to derive $x_1^*$ and $x_2^*$ in the closed-loop to summarize all the users' choices based on their $\theta$'s.

\begin{proposition}\label{prop:TMC-2}
At the equilibrium of Stage II, user fraction of the 5G and WiFi networks (i.e., $x_1^*$ and $x_2^*$) are solutions to 
\begin{equation}
u_1(F^{-1}(x_1^*))=u_2(F^{-1}(x_1^*)), \;\; u_2(F^{-1}(x_1^*+x_2^*))=\bar{u}, \label{e7}
\end{equation}
if $x_1^*+x_2^* < 1$ (see Fig. \ref{fig4}(a)), and otherwise 
\begin{equation}
u_1(F^{-1}(x_1^*))=u_2(F^{-1}(x_1^*)), \;\; x_1^*+x_2^*=1,\label{e8}
\end{equation}
with $u_2(F^{-1}(x_1^*+x_2^*))\geq \bar{u}$ (see Fig. \ref{fig4}(b)), where $x_2^* \in (0, 1)$ is the largest among all the possible solutions.
\end{proposition}

Proposition~\ref{prop:TMC-2}'s result is consistent with Fig. \ref{fig4}. It is interesting to notice that if $x_1^*+x_2^*=1$, it is possible to have $u_2(F^{-1}(x_1^*+x_2^*))> \bar{u}$, telling that the WiFi operator generously leaves positive payoff to the user at the boundary $\theta=F^{-1}(x_1^*+x_2^*)=1$. The purpose for doing so is to use a low WiFi price to reduce $x_1^*$ in 5G-only service (to the left-handed side in Fig. 3(b)) and increase the WiFi's own market share $x_2^*$.

 To deliver clean engineering insights, in all the following we consider { users obtain the same mobile access profits as $V_1 = V_2$ and} users' congestion sensitivity to follow the uniform distribution (i.e., $\theta \sim U[0,1]$). Though more involved, our analysis can be extended to other continuous $\theta$ distributions such as truncated normal distribution. We will show more results numerically in Section~\ref{section:TMC-6}.

Note that even with $\theta\sim U[0, 1]$, it is still difficult to solve closed-form $x_1^*$ and $x_2^*$, as combining \eqref{e7} and \eqref{e8} returns high degree polynomial equations and there are many cases depending on the price combinations. For example, the following proposition shows a case of the 4th degree polynomial equation, which is already the simplest among all the 5 case analysis in Fig. \ref{bg}. 
  
\begin{figure}
\centerline{\includegraphics[scale=0.38]{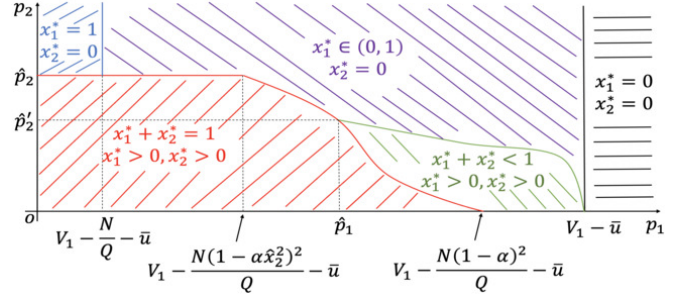}}
\caption{All 5 cases of users' subscription equilibrium ($x_1^*$, $x_2^*$) of the 5G and 5G+WiFi services versus the two operators' prices $p_1$ and $p_2$ in Stage II. } 
\label{bg} 
\end{figure}

\begin{proposition}\label{prop:TMC-3}
Figure \ref{bg} summarizes the structural results of the Stage II's equilibrium ($x_1^*$, $x_2^*$) in five cases, as functions of any given prices $p_1$ and $p_2$, where $\hat{p}_2$ and $\hat{p}_2'$ are given in \eqref{pro4_},
and $\hat{p}_1 \in [V_1 -\bar{u}- \frac{N}{Q}(1-\alpha\hat{x}_2^2)^2, V_1 - \bar{u}- \frac{N(1-\alpha)^2}{Q})$ is the unique solution to \eqref{pro4}.

For example, if both prices of the 5G and the add-on WiFi services are low (see the lower left region of Fig. \ref{bg}), i.e.,    
\begin{equation}
    p_1 \leq V_1 - \bar{u}- \frac{N}{Q}, \;\; p_2 \leq \alpha \hat{x}_2 \frac{N}{Q}\left(1-\alpha \hat{x}_2^2\right)(1-\hat{x}_2), \nonumber
\end{equation}
with $\hat{x}_2 \in [0, 1]$ as the unique solution to
\begin{equation}
    4\alpha \hat{x}_2^3 - 3\alpha \hat{x}_2^2 - 2\hat{x}_2 + 1 = 0, \label{hatx}
\end{equation}
then all the users either choose 5G or 5G+WiFi service without leaving the wireless market (i.e., $x_1^*+x_2^*=1$), and the equilibrium user fraction $x_2^*$ of 5G+WiFi is the greatest root in the range (0,1) to  
\begin{equation}
   p_2 - \alpha x_2^* \frac{N}{Q} \left(1-\alpha (x_2^*)^2\right) (1 - x_2^* ) = 0.\label{eq1}
\end{equation}
In another region (in the lower mid region of Fig. \ref{bg}), where the 5G capacity is small (i.e, $Q < \frac{N}{V_1-\bar{u}}$) and the WiFi coverage addition per AP is small (i.e., $\alpha < 1 - \sqrt{\frac{V_1-\bar{u}}{N/Q}}$), at the equilibrium not all the users choose 5G or 5G+WiFi service (i.e., $x_1^* + x_2^* < 1$).
\end{proposition}

We can observe from Fig. \ref{bg} that if $p_1$ or $p_2$ is too high, each network will lose subscribers, and Proposition~\ref{prop:TMC-3} tells that if both the 5G and WiFi networks charge low prices, they together occupy the whole market. 

Furthermore, the decreasing boundary curve (in red or green) in Fig. \ref{bg} implies that even if the WiFi operator wants to charge a higher price $p_2$, the 5G operator can still keep some users to add on the WiFi service to help offload by lowering his own price $p_1$.

\begin{corollary}\label{coro:TMC-1}
Assume the 5G operator keeps the same servcie price $p_1 = \bar{p}_1^*$  with $\bar{p}_1^*$ in \eqref{p_5} after the introduction of the crowdsourced WiFi network, there are a positive number of users to add on WiFi (i.e.,  $x_2^* > 0$) as long as $Q < \frac{3N}{V_1 -\bar{u}}$ and the WiFi deployment cost is small as 
\begin{align}
    &c \leq \frac{1 - \sqrt{1 - 3\alpha \sqrt{\frac{V_1-\bar{u}}{3N/Q}}}}{3} \frac{N}{Q} \sqrt{\frac{V_1-\bar{u}}{3N/Q}} \bigg(  \sqrt{\frac{V_1-\bar{u}}{3N/Q}} - \nonumber \\
    & \! \frac{1 - \sqrt{1 - 3\alpha \sqrt{\frac{V_1-\bar{u}}{3N/Q}}}}{3\alpha} \!\bigg(1 - \frac{1 - \sqrt{1 - 3 \alpha \sqrt{\frac{V_1-\bar{u}}{3N/Q}}}}{3}\bigg)\!\bigg).\!\! \!\label{c_eq}
\end{align}
\end{corollary}

\begin{proof}
See Appendix A. 
 \end{proof}

Intuitively, as the crowdsourced WiFi network's deployment cost $c$ is small or  the WiFi coverage addition per AP $\alpha$ increases (e.g., supported by the advanced amendments in 802.11ac/ad/ax) in \eqref{c_eq}, the WiFi operator's pricing becomes less-constrained or his network coverage improves, helping to attract a positive number of subscribers.   

{Note that we can run one-dimensional searches in the range of [0, 1] to solve polynomial equations (e.g., \eqref{eq1}) of Proposition~\ref{prop:TMC-3} for users' equilibrium choices in Stage II. The complexity order is $\Theta(\frac{1}{\epsilon_0})$ with precision error $\epsilon_0$.}

\section{Equilibrium Analysis of Operators' Pricing After the introduction of the Crowdsourced WiFi}\label{section:TMC-5}
In Stage I, by predicting users' equilibrium subscription in Stage II (as analyzed in Section~\ref{section:TMC-4}), the 5G and the crowdsourced WiFi operators simultaneously decide prices $p_1 \geq 0$ and $p_2 \geq c$ to maximize their own profits in \eqref{e5} and \eqref{e6}, respectively.

However, as shown in \eqref{eq1} and Proposition~\ref{prop:TMC-3}, we do not have closed-form solutions of $x_1^*$ and $x_2^*$ in Stage II, and it is difficult to derive the closed-form profit objectives $\pi_1$ and $\pi_2$. We are even not sure if these profits are concave or not with respect to $p_1$ and $p_2$.   Despite of this, similar to the related literature (e.g., \cite{gao2013economics, joe2015sponsoring}), we successfully derive analytical results to answer the key questions in this section: whether the 5G operator benefits from the introduction of the crowdsourced WiFi and how he should adapt his pricing. Our tractable  analysis approach is to first keep the 5G operator's price unchanged after the introduction of the crowdsourced WiFi (i.e., $p_1=\bar{p}_1^*$ with $\bar{p}_1^*$ in \eqref{p_5}), and we will prove that in this special case of the 5G operation, the 5G operator can still gain from the introduction of the crowdsourced WiFi. Note that the special case of $p_1=\bar{p}_1^*$ makes users' equilibrium subscription easier to predict.   
Recent studies (e.g., \cite{wang2018cross, han2016backhaul}) have proposed methods to allow WiFi links to communicate with heterogeneous protocols to avoid interference.
As this type of technology helps avoid interference, the WiFi offloading efficiency will be enhanced and we expect greater 5G profit after the introduction of the crowdsourced WiFi.

\begin{lemma}\label{lemma:TMC-3}
At the equilibrium of the whole dynamic game, the 5G operator obtains at least the same profit after the introduction of the crowdsourced WiFi, i.e., $\pi_1^* \geq \bar{\pi}_1^*$ with $\bar{\pi}_1^*$ in \eqref{pi_5}. 
\end{lemma}

\begin{proof}
It is sufficient to prove $\pi_1^* \geq \bar{\pi}_1^*$ in the special case of the 5G operation with  $p_1 = \bar{p}_1^*$. We divide our analysis into the following three cases, depending on $(x_1^*, x_2^*)$. 

\textit{Case 1.} If $x_2^* = 0$, telling that no user adds on the WiFi option at the equilibrium, the 5G operator faces the same situation as before the introduction of the crowdsourced WiFi. Thus, his profit remains unchanged.

\textit{Case 2.} If $x_2^* > 0$, $x_1^* + x_2^*(1  - \alpha x_2^*) \leq \bar{x}_1^*$, we have $x_1^* < \bar{x}_1^*$ due to $x_2^*(1  - \alpha x_2^*) > 0$ in the second condition above. According to Lemma~\ref{lemma:TMC-2}, users with $\theta > x_1^*$ have $u_2(\theta) > u_1(\theta)$. Therefore, we have $u_2(\bar{x}_1^*) > u_1(\bar{x}_1^*)$. Since
\begin{equation}
   u_1(\bar{x}_1^*)\! = \!V_1 \! -\!\frac{N}{Q}(x_1^* + x_2^*( 1 - \alpha x_2^* ))\bar{x}_1^* - \bar{p}_1^* \!
   \geq \!V_1 \!- \!\frac{N}{Q}(\bar{x}_1^*)^2 - \bar{p}_1^*, \label{a21}
\end{equation}
where the inequality holds due to $x_1^* + x_2^*(1  - \alpha x_2^*) \leq \bar{x}_1^*$. According to Lemma~\ref{lemma:TMC-1}, before the introduction of the crowdsourced WiFi, we have users with $\theta = \bar{x}_1^*$ obtain the fixed reservation payoff from joining the 5G network, that is,  
\begin{equation}
 \bar{u}_1(\bar{x}_1^*) = V_1 - \frac{N}{Q}(\bar{x}_1^*)^2 - \bar{p}_1^* = \bar{u}.   \label{a22}
\end{equation}
Combining \eqref{a21} and \eqref{a22}, we have $u_1(\bar{x}_1^*) \geq \bar{u}_1(\bar{x}_1^*) = \bar{u}$. Since $u_2(\bar{x}_1^*) > u_1(\bar{x}_1^*)$, we have $u_2(\bar{x}_1^*) > \bar{u}$.
If $\bar{x}_1^* = 1$, we have $u_2(1) > \bar{u}$. Since $u_2(\theta)$ in \eqref{e4} is decreasing in $\theta$, we have $u_2(\theta) > \bar{u}$ for $\theta > x_1^*$ and thus have $x_1^* + x_2^* = \bar{x}_1^* = 1$. If $\bar{x}_1^* < 1$ and $x_1^* + x_2^* = 1$, we have $x_1^* + x_2^* > \bar{x}_1^*$. If $\bar{x}_1^* < 1$ and $x_1^* + x_2^* < 1$,  we further notice that $u_2(x_1^* + x_2^*) = \bar{u}$ with $x_1^* + x_2^* < 1$ according to \eqref{e7}. Since $u_2(\theta)$ in \eqref{e4} is decreasing in $\theta$ and $u_2(\bar{x}_1^*) > \bar{u}$, we have $x_1^* + x_2^* > \bar{x}_1^*$. 
Then we conclude $ x_1^* + x_2^* \geq \bar{x}_1^*$. This tells that the 5G operator's subscriber number non-decreases at the same price $p_1=\bar{p}_1^*$, leading to at least the same profit for the 5G operator. 
 
\textit{Case 3.} If $x_2^* > 0$ and $x_1^* + x_2^*(1  - \alpha x_2^*)  > \bar{x}_1^*$, we have $ x_1^* + x_2^* > \bar{x}_1^*$ due to $1  - \alpha x_2^* < 1$. Therefore, the 5G operator's subscriber number increases at the same price $p_1=\bar{p}_1^*$, leading to a greater profit.
 \end{proof}

Following Lemma~\ref{lemma:TMC-3}, we further analyze when the 5G operator strictly benefits from the introduction of the crowdsourced WiFi. 

\begin{proposition}\label{prop:TMC-4}
At the equilibrium of the whole dynamic game, the 5G operator obtains a strictly greater profit after the introduction of  the crowdsourced WiFi, i.e., $\pi_1^* > \bar{\pi}_1^*$ with $\bar{\pi}_1^*$ in \eqref{pi_5}, as long as the 5G capacity is small (i.e., $Q < 3N/(V_1 -\bar{u})$) and WiFi deployment cost is small in \eqref{c_eq}.
\end{proposition}

\begin{proof}
See Appendix B.
\end{proof}

Intuitively, Proposition~\ref{prop:TMC-4} tells that a positive fraction $x_2^*>0$ of users in the crowdsourced WiFi network will help offload 5G traffic to WiFi, and thus the 5G operator benefits from the introduction of the crowdsourced WiFi. 

By combining Corollary~\ref{coro:TMC-1} and Proposition~\ref{prop:TMC-4}, we have the result below.
\begin{corollary}\label{coro:TMC-2}
At the equilibrium, a positive number of users will add on WiFi (i.e., $x_2^* > 0$) as long as  $Q < \frac{3N}{V_1-\bar{u}}$  and the WiFi deployment cost $c$ is small in \eqref{c_eq}. In this case, the 5G operator strictly benefits from the introduction of the crowdsourced WiFi.
\end{corollary}

As the crowsourced WiFi is introduced to cater to the users with high congestion costs, it has positive user demand ($x_2^*>0$) if the 5G network is severely congested with $Q < \frac{3N}{V_1-\bar{u}}$. Further, the WiFi network should be affordable to deploy, requiring \eqref{c_eq} implying small $c$ or large $\alpha$  to hold.

Next we want to analyze how the 5G operator should adapt his pricing to the crowdsourced WiFi. Recall from Section~\ref{section:TMC-4} that we cannot obtain the closed-form expressions  of $x_1^*$ and $x_2^*$ as functions of $p_1$ and $p_2$. Alternatively, under the target case of $x_2^*>0$, we simplify the profit objective of the 5G operator to
\begin{align}
\pi_1(p_1, x_2^*) = \;\;\;\;\;\;\;\;\;\;\;\;\;\;\;\;\;\;\;\;\;\;\;\;\;\;\;\;\;\;\;\;\;\;\;\;\;\;\;\;\;\;\;\;\;\;\;\;\;\;\;\;\;\;\;\;\;\;\;\; \nonumber\\
  \begin{cases}
  Np_1,  &\!\!\!\!\text{if}\; p_1\! \leq\! V_1\!\! -\bar{u}-\!\! \frac{N}{Q}(1-\alpha \bar{x}_2^2)^2 ,   \\
  N p_1 \!\big( \!\sqrt{\frac{V_1-\bar{u}-p_1}{N/Q}} + \alpha (x_2^*)^2 \!\big)\!, &\!\!\!\!\text{if}\;V_1 \!- \bar{u}-\!\frac{N}{Q}(1-\bar{x}_2^2)^2 \!< \!p_1  \\
  &\!\!\!\!< \!\! V_1 - \bar{u}-\frac{N}{Q}(\bar{x}_2')^2(1- \alpha \bar{x}_2')^2, 
\end{cases} \label{a42}  
\end{align}
where $x_2^* = \max\{\bar{x}_2, \bar{x}_2'\}$, $\bar{x}_2 \in [\hat{x}_2, 1)$ is the greatest root to
\begin{equation*}
    \alpha \bar{x}_2 \frac{N}{Q} (1-\alpha \bar{x}_2^2)(1-\bar{x}_2) - p_2 = 0,
\end{equation*}
$\hat{x}_2$ is the unique solution to \eqref{hatx}
and $\bar{x}_2' \in (0, 1)$ is the greatest root to
\begin{equation*}
    \alpha \bar{x}_2' \frac{N}{Q} \sqrt{\frac{V_1-\bar{u}-p_1}{N/Q}}\bigg(\sqrt{\frac{V_1-\bar{u}-p_1}{N/Q}}-\bar{x}_2' + \alpha (\bar{x}_2')^2\bigg) \!- \!p_2\! =\! 0.
\end{equation*}

Note that it is still difficult to derive the structural result to decide the best pricing response for the 5G operator to the WiFi network, and the profit objective or pricing response of the WiFi operator, not to mention solving the final equilibrium prices ($p_1^*$, $p_2^*$). Despite of this difficulty, we analytically derive the following structural result for guiding the 5G operation.  

\begin{table*}
   \begin{align}
    c \in \Bigg( &\frac{1 - \sqrt{1 - 3\alpha \sqrt{\frac{V_1-\bar{u}}{3N/Q}}}}{3} \frac{N}{Q} \sqrt{\frac{V_1-\bar{u}}{3N/Q}} \bigg(  \sqrt{\frac{V_1-\bar{u}}{3N/Q}} - 
    \frac{1 - \sqrt{1 - 3\alpha \sqrt{\frac{V_1-\bar{u}}{3N/Q}}}}{3\alpha} \!\bigg(1 - \frac{1 - \sqrt{1 - 3 \alpha \sqrt{\frac{V_1-\bar{u}}{3N/Q}}}}{3}\bigg)\bigg),\nonumber \\
    &
    \frac{1 - \sqrt{1 - 3\alpha \sqrt{\frac{V_1-\bar{u}}{N/Q}}}}{3} \frac{N}{Q} \sqrt{\frac{V_1-\bar{u}}{N/Q}} \bigg(  \sqrt{\frac{V_1-\bar{u}}{N/Q}} - 
   \frac{1 - \sqrt{1 - 3\alpha \sqrt{\frac{V_1-\bar{u}}{N/Q}}}}{3\alpha} \!\bigg(1 - \frac{1 - \sqrt{1 - 3 \alpha \sqrt{\frac{V_1-\bar{u}}{N/Q}}}}{3}\bigg)\bigg)\!\Bigg).
    \label{by5}
\end{align}
\hrulefill
\end{table*}

\subsection{Large 5G capacity Regime with $Q \geq \frac{3N}{V_1-\bar{u}}$}\label{section:TMC-5.1}

\begin{proposition}\label{prop:TMC-5}
At the equilibrium, the 5G network operator charges at least the same 5G price (i.e., $p_1^* \geq \bar{p}_1^*$ with $\bar{p}_1^*$ in \eqref{p_5}) if the 5G capacity is large (i.e., $Q \geq 3N/(V_1-\bar{u})$). Furthermore, the 5G operator charges a strictly greater price (i.e., $p_1^* > \bar{p}_1^*$) if the 5G capacity is large (i.e., $Q\geq 3N/(V_1-\bar{u})$) and the WiFi deployment cost is small, i.e.,
\begin{equation}
    c \leq \alpha\hat{x}_2\frac{N}{Q}(1-\alpha\hat{x}_2^2)(1-\hat{x}_2), \label{by8}
\end{equation}
with $\hat{x}_2 \in (0, 1)$ as the unique solution to
\begin{equation*}
    4\alpha \hat{x}_2^3 - 3\alpha \hat{x}_2^2 - 2\hat{x}_2 + 1 = 0.
\end{equation*}
\end{proposition}
\begin{proof}
See Appendix C.
\end{proof}

Intuitively, given large 5G capacity, the 5G network is not severe in congestion and already serves all the users (see $\bar{x}_1^*=1$ in \eqref{x_5} in Proposition~\ref{prop:TMC-1}). After the introduction of  the crowdsourced WiFi with $x_2^*$ fraction of users to offload, the 5G congestion reduces and the 5G operator wants to charge a higher price due to improved service quality.

With Proposition~\ref{prop:TMC-5}, we can directly give the condition when the 5G operator strictly benefits from the introduction of the crowdsourced WiFi below.

\begin{figure*}
\centering
\subfloat[When the 5G network capacity is small, $Q = 30$]{\includegraphics[scale = 0.3]{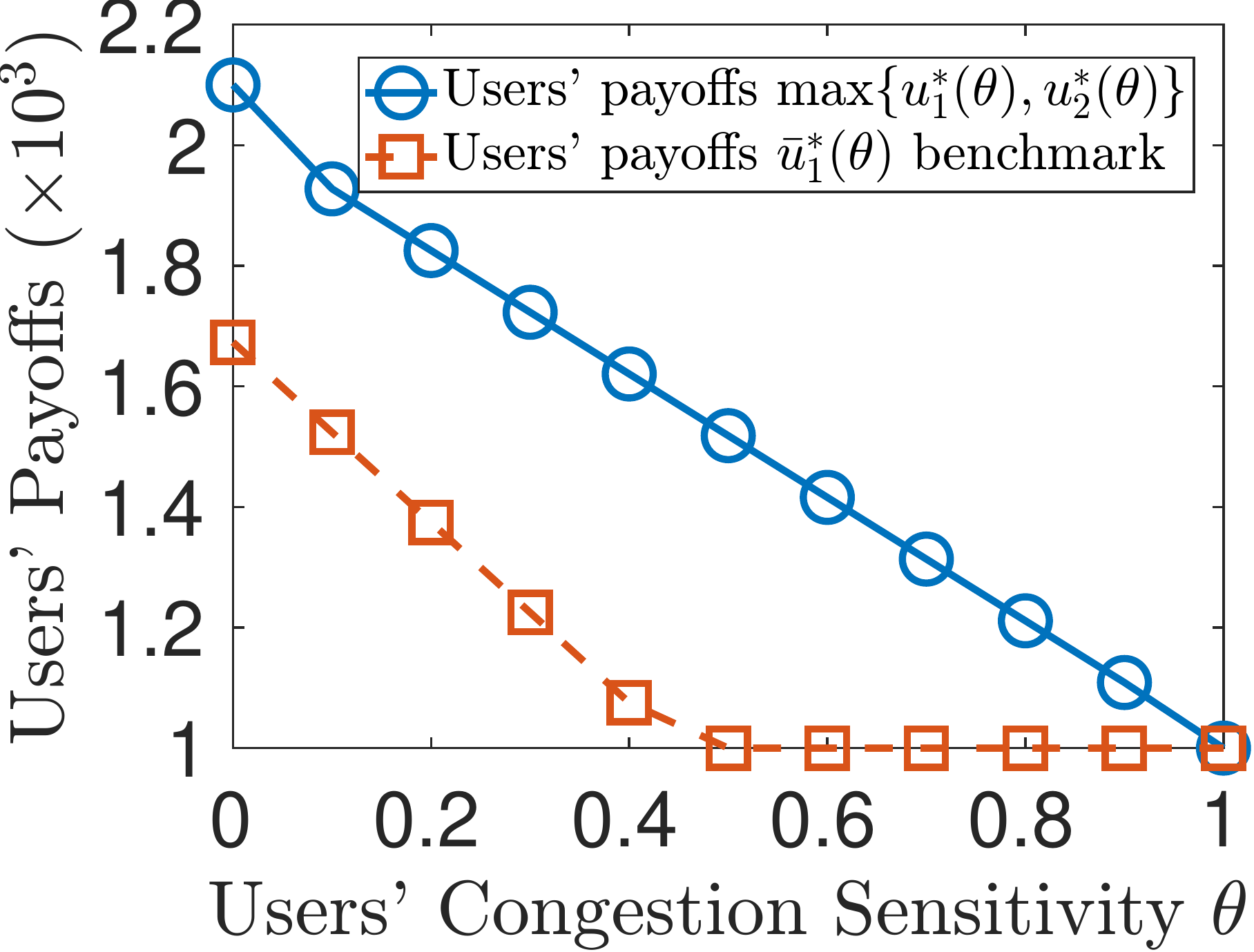}%
\label{fig_first_case}}
\hfil
\subfloat[When the 5G network capacity is small, $Q = 120$]{\includegraphics[scale = 0.3]{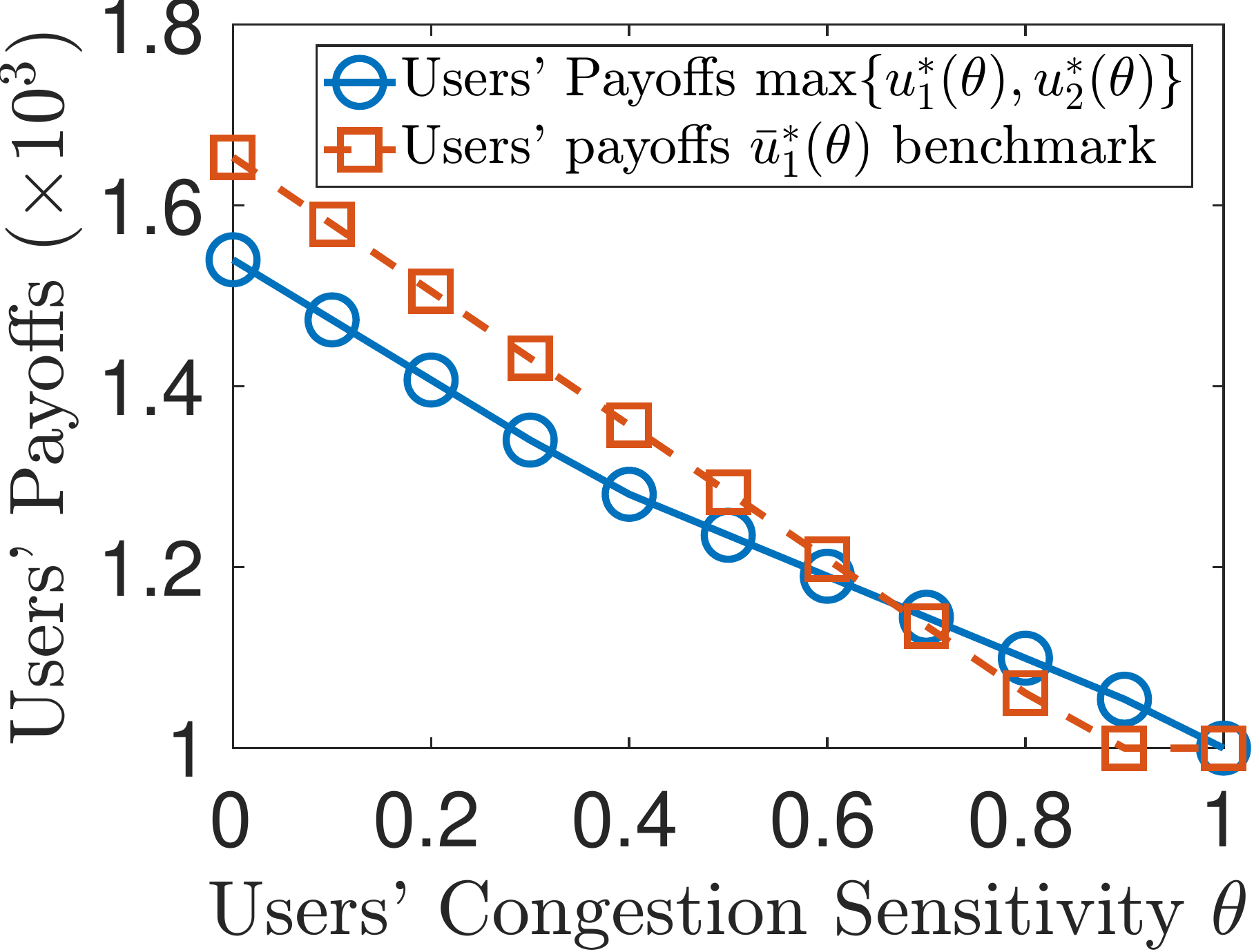}%
\label{fig_second_case}}
\hfil
\subfloat[When the 5G network capacity is non-small, $Q = 180$]{\includegraphics[scale = 0.3]{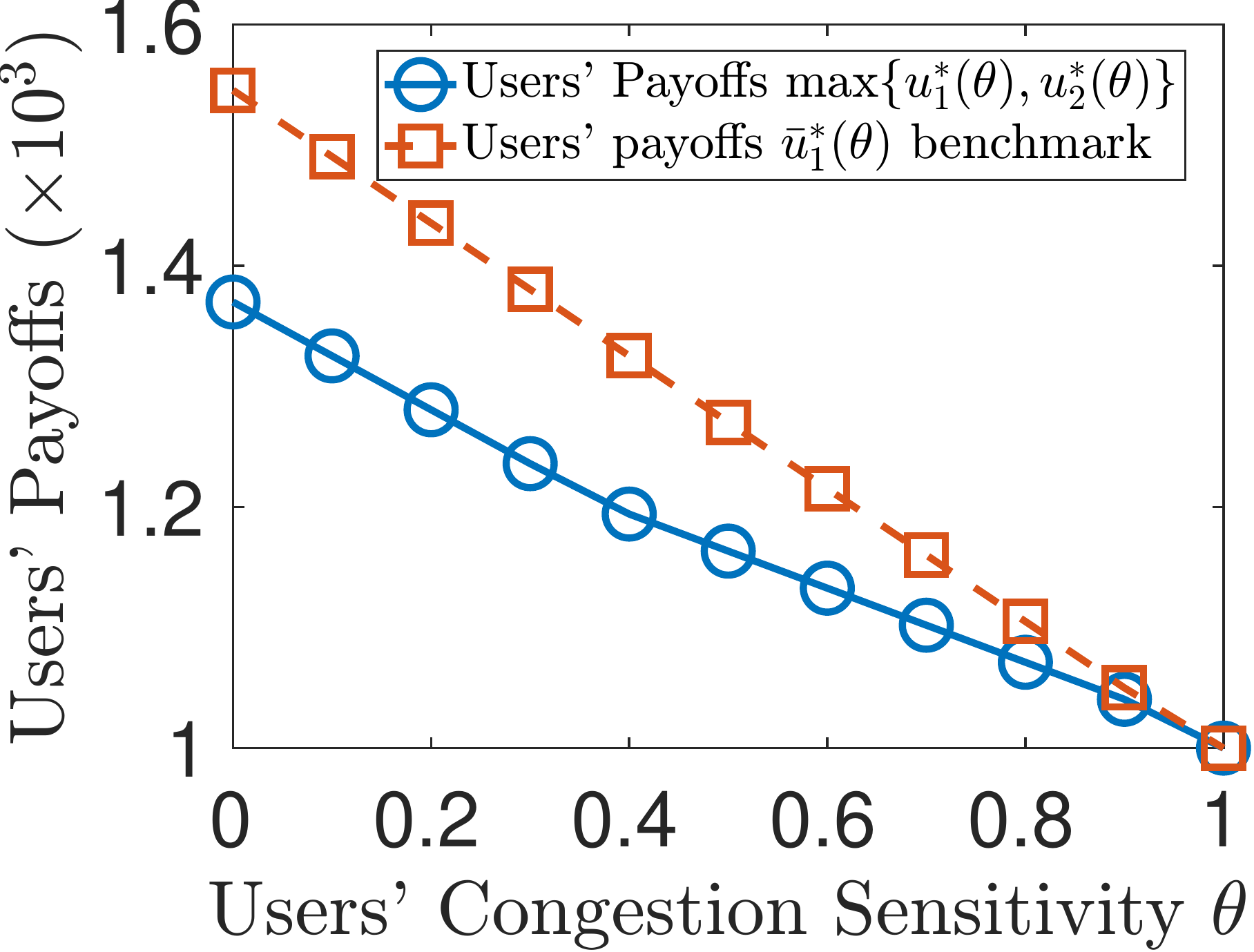}%
\label{fig_third_case}}
\caption{Comparison between all the users' equilibrium payoffs $\max\{u_1^*(\theta), u_2^*(\theta)\}$ and benchmark $\bar{u}_1^*(\theta)$ with and without the crowdsourced WiFi. Here, a user's congestion sensitivity $\theta$ follows truncated distribution $N(0.5, 1)$ in the normalized range [0, 1]. We vary  congestion sensitivities $\theta$, and fix $N = 10^5$, {$V_1 = 3000$}, $\bar{u} = 1000$, $c = 50$, and $\alpha = 0.5$.}
\label{fig-re1}
\end{figure*}

\begin{proposition}\label{prop:TMC-6}
At the equilibrium of the whole dynamic game, the 5G operator obtains a strictly greater profit after the introduction of  the crowdsourced WiFi, i.e., $\pi_1^* > \bar{\pi}_1^*$ with $\bar{\pi}_1^*$ in \eqref{pi_5}, as long as the WiFi deployment cost is small in \eqref{by8},
and the 5G network capacity is large (i.e., $Q \geq 3N/(V_1 -\bar{u})$).
\end{proposition}

Perhaps surprisingly, we prove that users' payoffs may worsen after the introduction of the crowdsourced WiFi, if the 5G capacity is large.
\begin{proposition}\label{prop:TMC-7}
At the equilibrium, all the users obtain strictly less payoff after the introduction of  the crowdsourced WiFi  (i.e, $\max\{u_1^*(\theta), u_2^*(\theta)\}<\bar{u}_1^*(\theta)$), if the 5G capacity is non-small (i.e., $Q\geq 3N/(V_1-\bar{u})$) and WiFi deployment cost is small in \eqref{by8}.  
\end{proposition}

\begin{proof}
See Appendix D.
\end{proof}

\subsection{Medium and Small 5G capacity Regimes with $Q < \frac{3N}{V_1-\bar{u}}$}\label{section:TMC-5.2}

 The following proposition further extends Proposition~\ref{prop:TMC-5}'s result to hold if $Q$ if slightly smaller than $3N/(V_1-\bar{u})$. 
\begin{proposition}\label{prop:TMC-8}
At the equilibrium, the 5G operator charges users more (i.e., $p_1^* > \bar{p}_1^*$) after the introduction of the crowdsourced WiFi, if the 5G network capacity is medium, i.e.,
\begin{equation*}
    3\left(\frac{\sqrt{17}+23}{32}\right)^2\frac{N}{V_1-\bar{u}} < Q < \frac{3N}{V_1-\bar{u}},
\end{equation*}
the WiFi coverage addition per AP is non-small, i.e.,
\begin{equation}
    \alpha > \bigg(\frac{2\sqrt{1-\sqrt{\frac{V_1-\bar{u}}{3N/Q}}}\bigg(1-2\sqrt{\frac{V_1-\bar{u}}{3N/Q}}\bigg)}{2-3\sqrt{\frac{V_1-\bar{u}}{3N/Q}}}\bigg)^2, \label{alpha}
\end{equation}
and the WiFi deployment cost is small in \eqref{by8}.
\end{proposition}

\begin{proof}
See Appendix E.
\end{proof}

Intuitively, as compared to Proposition 5 with large capacity, in Proposition 8 the 5G capacity becomes less sufficient and the 5G operator further expects non-small WiFi coverage in \eqref{alpha} to help offload and reduce 5G congestion. 

Finally, we turn to the small capacity regime for the 5G network. Unlike Propositions~\ref{prop:TMC-5} and \ref{prop:TMC-8}, the 5G operator with small capacity will no longer charge a higher 5G service price from users after the introduction of the crowdsourced WiFi. 
\begin{proposition}\label{prop:TMC-9}
At the equilibrium, the 5G operator charges users less (i.e., $p_1^* < \bar{p}_1^*$) after the introduction of the crowdsourced WiFi, if the 5G network capacity is small (i.e., $Q < N/(V_1-\bar{u})$), the WiFi coverage addition per AP is small (i.e., $\alpha < 1 - \sqrt{\frac{V_1-\bar{u}}{N/Q}}$), and
 the   WiFi deployment cost is non-small in \eqref{by5}.
\end{proposition}

\begin{proof}See Appendix F. \end{proof}

Given small capacity, the 5G network faces severe congestion. To motivate users to add on the WiFi service for WiFi offloading, the 5G operator will purposely lower his 5G price to his users and reduce the total 5G+WiFi service payment. 

One may be aware that there is a gap between the small and medium regimes of 5G capacity in the analytical results of Propositions~\ref{prop:TMC-8} and \ref{prop:TMC-9}. We will make up this gap to show similar observations using numerical results in next section. 

{Note that we can find operators' pricing equilibrium ($p_1^*$, $p_2^*$) by combining 5G and WiFi's best-response price functions $p_1^*(p_2)$ and $p_2^*(p_1)$. We run one-dimensional searches to find $p_1^*(p_2)$ and $p_2^*(p_1)$, respectively, given users' participation equilibrium ($x_1^*(p_1, p_2)$, $x_2^*(p_1,p_2)$) in Stage II. There we choose precision error $\epsilon_1$ to search 5G's best-response price $p_1^*(p_2)$ in the range of [0, $V_1$], and $\epsilon_2$ to search WiFi's best-response price $p_2^*(p_1)$ in the range of [0, $V_2$]. For each pair of ($p_1$, $p_2$), it involves complexity $\Theta(\frac{1}{\epsilon_0})$ to find users' participation equilibrium ($x_1^*(p_1, p_2)$, $x_2^*(p_1,p_2)$) in Stage II. To obtain discrete functions $p_1^*(p_2)$ or $p_2^*(p_1)$, it requires  
complexity order $\Theta(\frac{V_1 V_2}{\epsilon_0 \epsilon_1 \epsilon_2})$. To obtain the whole equilibrium of the two-stage game, we easily combine $p_1^*(p_2)$ and $p_2^*(p_1)$ to find their interaction and the complexity is also $\Theta(\frac{V_1 V_2}{\epsilon_0 \epsilon_1 \epsilon_2})$. }

\section{Numerical Results}\label{section:TMC-6}

In this section, we look at the whole picture and present numerical results to tell how the introduction of the crowdsourced WiFi affects the 5G service pricing/profit and the users' payoffs at the equilibrium.  Recall that we focus on uniform distribution for users' congestion sensitivity in Section~\ref{section:TMC-5}, in this section we relax all the figures' setting to consider truncated normal distributions to numerically show similar insights. {The WiFi deployment cost per access point includes the special router cost in tens or hundreds of dollars, and we reasonably set $c=100$ and 50.}

Given this relaxation\footnote{As the mean of users' sensitivity distribution increases, we expect more users will add on the crowdsourced WiFi  to mitigate users' congestion sensitivities on average; as the variance of users' sensitivity distribution increases, we expect more users will add on WiFi to alleviate congestion in the 5G network.}, Figure \ref{fig-re1} shows all users' equilibrium payoffs $\max\{u_1^*(\theta)$, $ u_2^*(\theta)\}$ versus their 5G congestion sensitivity $\theta$ and compares with the benchmark $\bar{u}_1^*(\theta)$ before the introduction of the crowdsourced WiFi. Note that in Fig. \ref{fig_first_case} all the users' payoffs improve after the introduction of the crowdsourced WiFi due to 5G operator's less charging. As the 5G capacity improves in Figures \ref{fig_second_case}-\ref{fig_third_case}, users with small congestion sensitivity do not gain much from WiFi offloading to reduce congestion and their payoffs worsen due to 5G operator's over-charging after the introduction of the crowdsourced WiFi;  as the 5G capacity becomes non-small in Fig. \ref{fig_third_case}, this effect spreads to all the users and their payoffs all worsen, which is consistent with Proposition~\ref{prop:TMC-7}.

\begin{figure}
\centerline{\includegraphics[scale=0.28]{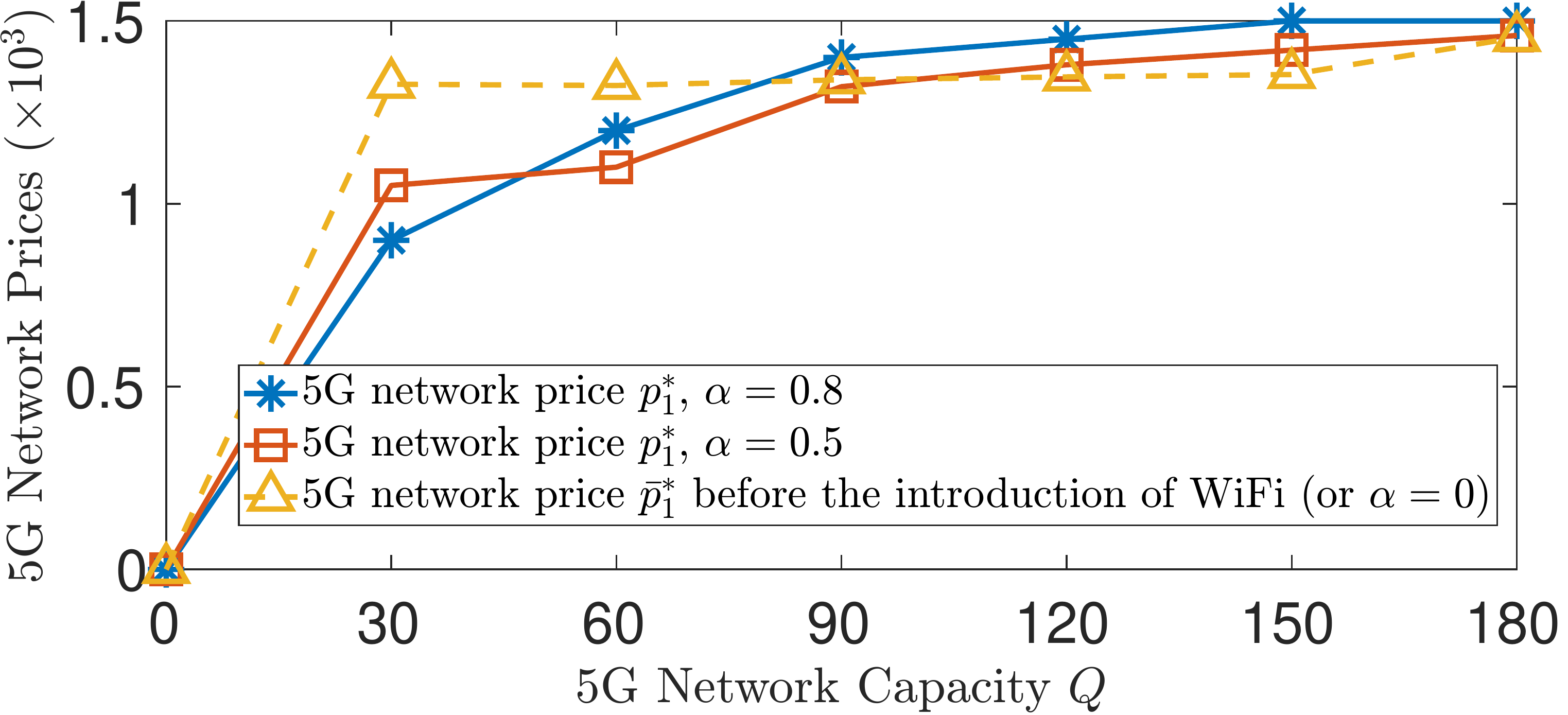}}
\caption{5G network's annual prices $\bar{p}_1^*$ and $p_1^*$ before and after the introduction of the crowdsourced WiFi, versus the 5G capacity $Q$ and the WiFi coverage addition per AP $\alpha$, respectively. Here, a user's congestion sensitivity $\theta$ follows truncated distribution $N(0.5, 1)$ in the normalized range [0, 1]. We fix $N = 10^5$, { $V_1 = 3000$}, $\bar{u} = 1000$, and {$c = 100$}. Note that the benchmark $\bar{p}_1^*$ before the introduction of the crowdsourced WiFi is equivalent to the crowdsourced WiFi case with $\alpha=0$.}
\label{fig10011}

\end{figure}

\begin{figure}
\centerline{\includegraphics[scale=0.28]{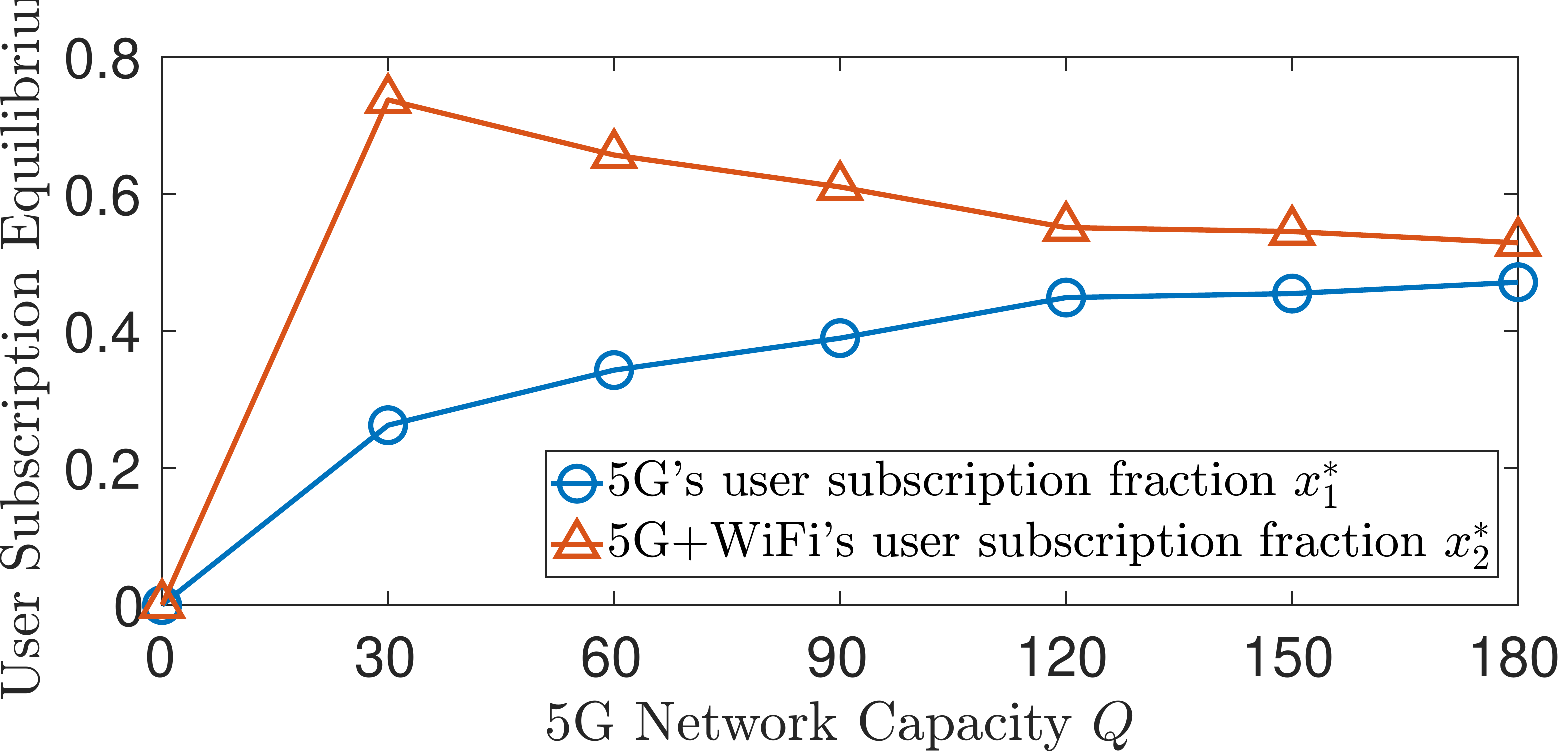}}
\caption{User subscriptions to the 5G service only $x_1^*$ and the 5G+WiFi $x_2^*$ at the equilibrium versus the 5G network capacity $Q$. Here we fix $N = 10^5$, {$V_1 = 3000$}, $\bar{u} = 1000$, {$c = 100$}, and $\alpha = 0.8$.}  
\label{fig31}
\end{figure}

Figure \ref{fig10011} shows the 5G network's equilibrium prices $\bar{p}_1^*$ and $p_1^*$, before and after the introduction of  the crowdsourced WiFi, as functions of the 5G capacity $Q$ and WiFi coverage addition $\alpha$, respectively. { Here we refer the 5G operator's prices as users' annual payment to network operators. Such payment ranges from hundreds to thousands of dollars\footnote{AT\&T's 5G plan, https://www.att.com/plans/unlimited-data-plans/.}.} We observe that $p_1^*$ becomes greater than $\bar{p}_1^*$ if the 5G network capacity $Q$ is no less than $3N/(V_1-\bar{u}) = 150$, which is consistent with Proposition~\ref{prop:TMC-5}. The over-charging pricing holds for slightly smaller $Q$, which echos with Proposition~\ref{prop:TMC-8}.

 Figure \ref{fig10011} also shows that the 5G network price $p_1^*$ increases as the network capacity $Q$ increases, and is less than benchmark $\bar{p}_1^*$ if $Q \leq 60$, which is consistent with Proposition~\ref{prop:TMC-9}. For this small 5G capacity case, many users face severe congestion. To motivate more users to add on WiFi option and help offload heavy 5G traffic, the 5G operator purposely lowers his price than $\bar{p}_1^*$ as a compensation to users.   

 Moreover, Fig. \ref{fig10011} shows that when the 5G network capacity is small (e.g., $Q = 30$), the 5G operator may decrease his price as the WiFi coverage addition per AP   $\alpha$ increases from 0 to 0.5 and then 0.8. The intuition is that as $\alpha$ increases, the 5G operator can more efficiently offload the traffic to WiFi and reduce his own network congestion. Thus, he chooses to further lower his price and motivate more users to add-on WiFi.

Figure \ref{fig31} shows user subscription equilibrium $x_1^*$ and $x_2^*$ in the 5G and 5G+WiFi services  as functions of $Q$, respectively. We observe that $x_1^*$ increases with $Q$ and will converge as $Q$ becomes large.

 \begin{figure}
\centerline{\includegraphics[scale=0.28]{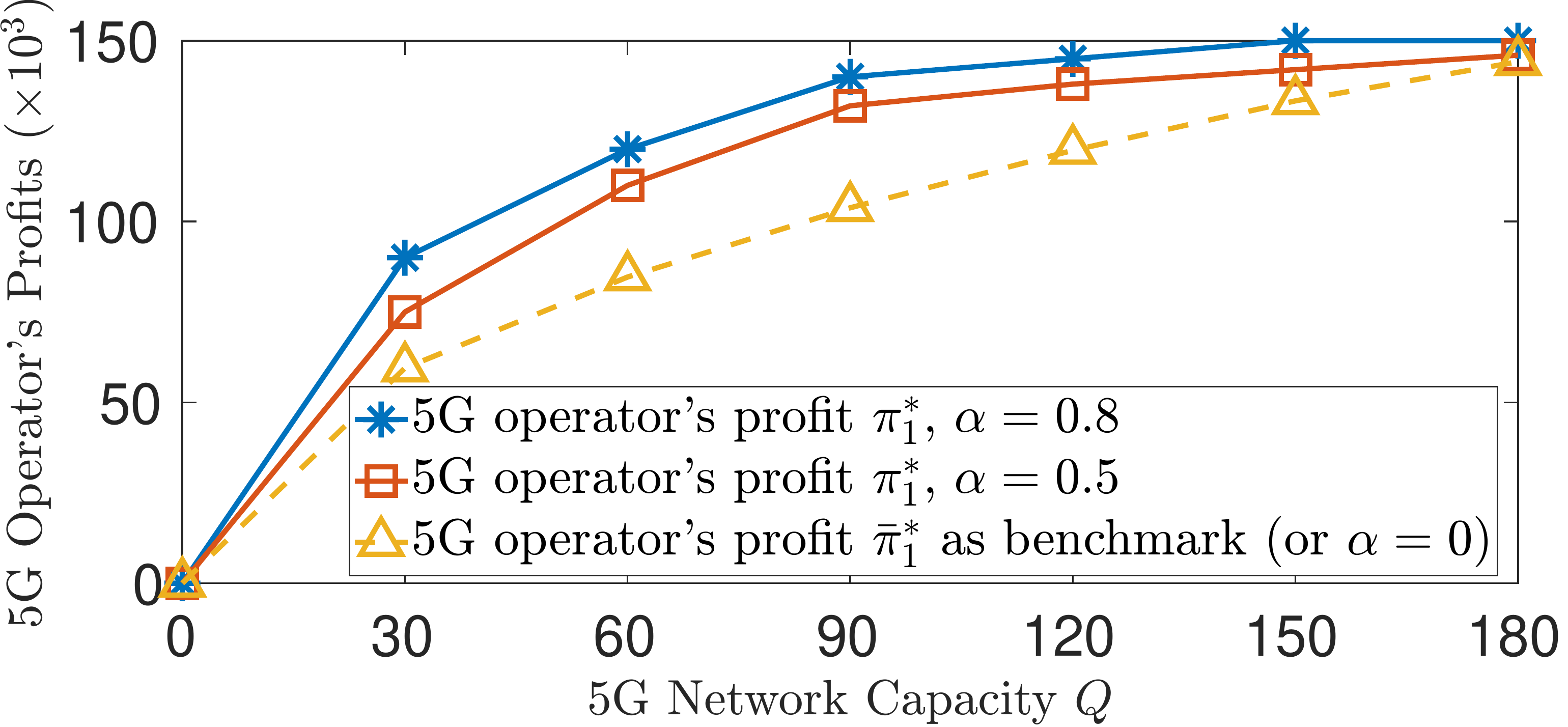}}
\caption{5G operator's profits $\bar{\pi}_1^*$ and $\pi_1^*$ before and after the introduction of the crowdsourced WiFi, versus the 5G capacity $Q$ and the WiFi coverage addition per AP $\alpha$, respectively. Here, a
user's congestion sensitivity $\theta$ follows truncated distribution $N(0.5, 1)$ in the normalized range [0, 1]. We fix $N = 10^5$, {$V_1 = 3000$}, $\bar{u} = 1000$, and {$c = 100$}. Note that the benchmark $\bar{\pi}_1^*$ before the introduction of the crowdsourced WiFi is equivalent to the crowdsourced WiFi case with $\alpha=0$.}
\label{fig1100}
\end{figure}

It is interesting to notice from Fig. \ref{fig31} that the user fraction $x_2^*$ in 5G+WiFi service first increases and then decreases with $Q$. When $Q$ is small, many users still face severe 5G congestion outside the WiFi coverage and choosing the complementary WiFi service to 5G does not lead to a positive payoff. As $Q$ increases, the 5G congestion mitigates and $x_2^*$ for the combo 5G+WiFi service increases with $Q$. When $Q$ becomes non-small, the 5G service has low congestion and users will choose downgrade from 5G+WiFi to 5G service only, for saving the payment to the WiFi service.

Figure \ref{fig1100} shows the 5G operator's equilibrium profits $\bar{\pi}_1^*$ and $\pi_1^*$ before and after the introduction of  the crowdsourced WiFi, respectively, as functions of $Q$ and the WiFi coverage addition per AP $\alpha$. All profits are increasing in $Q$ given more capacity resources. The introduction of the crowdsourced WiFi helps offload the 5G traffic (though provided by another self-interested WiFi operator with/without congestion), and the 5G profit $\pi_1^*$ increases with the WiFi coverage addition $\alpha$. Actually, the benchmark case before the introduction of the crowdsourced WiFi can be viewed as the latter case with $\alpha=0$. We also notice from Fig. \ref{fig1100} that the profit improvement after the introduction of  the crowdsourced WiFi first increases and then decreases with $Q$. This is consistent with Fig. \ref{fig31}, where $x_2^*$ to help offload first increases and then decreases with $Q$.

{ We define social welfare after the introduction of the crowdsourced WiFi as the sum of the 5G, crowdsourced WiFi's profits and $N$ users' payoffs, i.e., 
\begin{equation*}
    SW = \pi_1^* + \pi_2^* + N \int_{\theta = 0}^1 \max\{u_1^*(\theta), u_2^*(\theta), \bar{u}\} dF(\theta).
\end{equation*}
Before the introduction of the crowdsourced WiFi, social welfare $\Bar{SW}$ is
\begin{equation*}
    \Bar{SW} = \bar{\pi}_1^* + N \int_{\theta = 0}^1 \max\{\bar{u}^*_1(\theta), \bar{u}\} dF(\theta).
\end{equation*}

 \begin{figure}
\centerline{\includegraphics[scale=0.28]{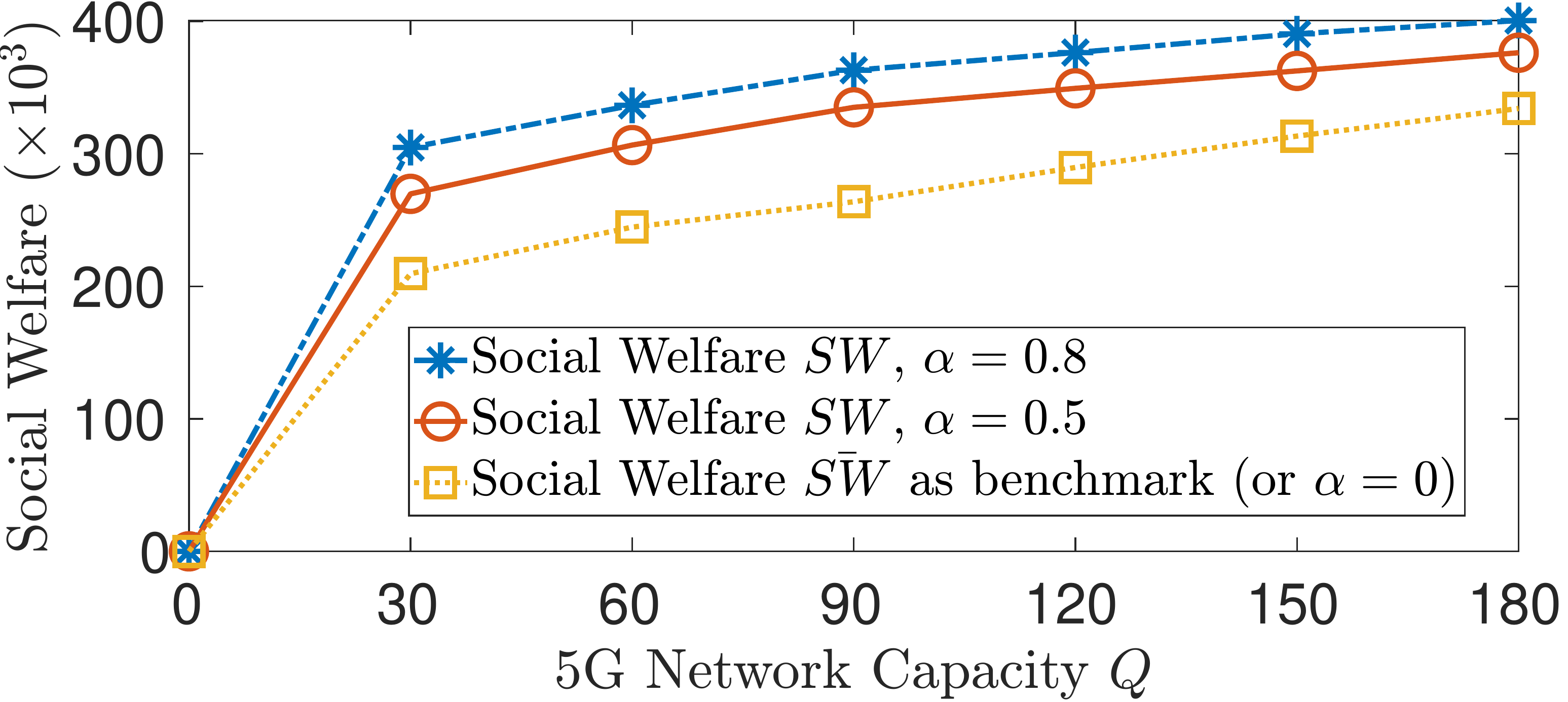}}
\caption{ Social welfare $\bar{SW}$ and $SW$ before and after the introduction of the crowdsourced WiFi, with user congestion sensitivity $\theta$ following truncated distribution $N(0.5, 1)$ in the normalized range [0, 1]. Here we vary the 5G capacity $Q$ and the WiFi coverage addition per AP $\alpha$, and fix $N = 10^5$, $V_1 = 3000$, $\bar{u} = 1000$, and {$c = 100$}. Note that the benchmark $\bar{\pi}_1^*$ before the introduction of the crowdsourced WiFi is equivalent to the crowdsourced WiFi case with $\alpha=0$.}
\label{fig101010}
\end{figure}

 Figure \ref{fig101010} shows social welfare before and after the introduction of  the crowdsourced WiFi, respectively, as functions of $Q$ and the WiFi coverage addition per AP $\alpha$, which has the similar pattern as in Fig. \ref{fig1100} to increase with $\alpha$ and $Q$, respectively.}
 
 \begin{figure}
\centerline{\includegraphics[scale=0.28]{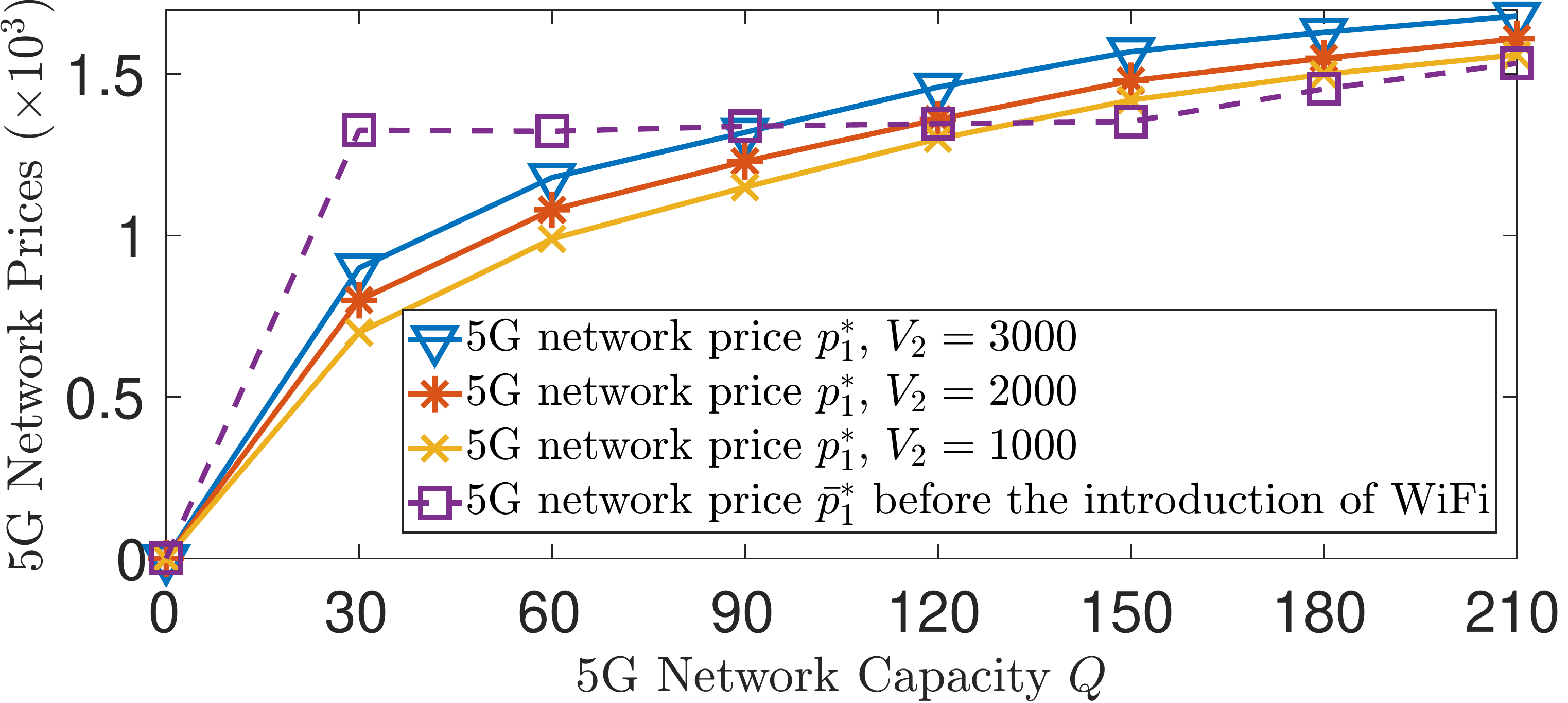}}
\caption{5G network's annual prices $\bar{p}_1^*$ and $p_1^*$ before and after the introduction of the crowdsourced WiFi versus the 5G capacity $Q$ and the WiFi mobile access benefit $V_2$, respectively.
Here, a user's congestion sensitivity $\theta$ follows truncated distribution $N(0.5, 1)$ in the normalized range [0, 1]. And fix $N = 10^5$,  $V_1 = 3000$,  $\bar{u} = 1000$, $\alpha = 0.5$, and {$c = 50$}. }
\label{figa}
\end{figure}

{Figure \ref{figa} shows the 5G network's equilibrium prices $\bar{p}_1^*$ and $p_1^*$, before and after the introduction of  the crowdsourced WiFi. We present prices as functions of the 5G capacity $Q$ and WiFi service benefit $V_2$ given the same $V_1$ for 5G, respectively. Though $V_2$ is different from $V_1$, we find that our key insights from Propositions 5 and 8-9 still hold.  That is, the 5G operator charges less than benchmark $\bar{p}_1^*$ for small capacity $Q$ as many users face severe congestion. The smaller 5G price helps motivate more users to add on WiFi option and offload heavy 5G traffic. With non-small capacity, the 5G network is less congested and the 5G operator is thus able to charge higher price for more profit.
As $V_2$ increases, to add-on WiFi gives users greater payoffs in \eqref{e4} and it becomes easier for the 5G operator to motivate WiFi offloading. Thus,  Fig. \ref{figa} shows that the 5G operator is able  to charge more in 5G service. 

Figure \ref{figb} shows the 5G operator's equilibrium profits $\bar{\pi}_1^*$ and $\pi_1^*$, before and after the introduction of the crowdsourced WiFi. We can tell from this figure that though $V_2$ is different from $V_1$, our key insights from Propositions 4 and 6 still hold. That is, the 5G operator obtains more profit than benchmark $\bar{\pi}_1^*$ in Fig. \ref{figb} because the introduction of the crowdsourced WiFi helps offload the 5G traffic. As $V_2$ increases,  more users will add-on WiFi to help WiFi offloading. Thus,  Fig. \ref{figb} shows that the 5G operator is able to make more profits.        }

\section{Extension to congested WiFi case}\label{section:TMC-7}

In this section we further consider congestion in the crowdsourced WiFi network as in the 5G network. A 5G-only user's payoff is still \eqref{e3}, yet a 5G+WiFi user's is no longer \eqref{e4}. When moving spatially, he finds the WiFi coverage with probability $\alpha x_2$ and experiences congestion cost $\beta$ inside\footnote{The congestion cost $\beta$ may not scale up with the user fraction $x_2$ to add-on WiFi. As $x_2$ increases, more users deploy APs to contribute to the crowdsourced WiFi network, making the average user demand per AP unchanged.}. His expected WiFi congestion cost is thus $\alpha x_2 \beta \theta$ with individual congestion sensitivity $\theta$. By adding this extra cost to \eqref{e4}, his payoff becomes: 
\begin{equation}
    \tilde{u}_2(\theta) \!\!=\!\! V  \!- (1 -\alpha x_2)\frac{N(x_1+x_2(1 - \alpha x_2))}{Q}\theta - \alpha x_2 \beta \theta  - p_1 - p_2. \label{new_u2}
\end{equation}

For ease of exposition, we assume $V_1=V_2=V$ here.
After considering the WiFi congestion, we still face the same two-stage dynamic game among the 5G operator, the crowdsourced WiFi and $N$ users in Section 2.2, by replacing \eqref{e4} with \eqref{new_u2}. We next use backward induction to analyze the game and show the possible impact of WiFi congestion.

 \begin{figure}
\centerline{\includegraphics[scale=0.28]{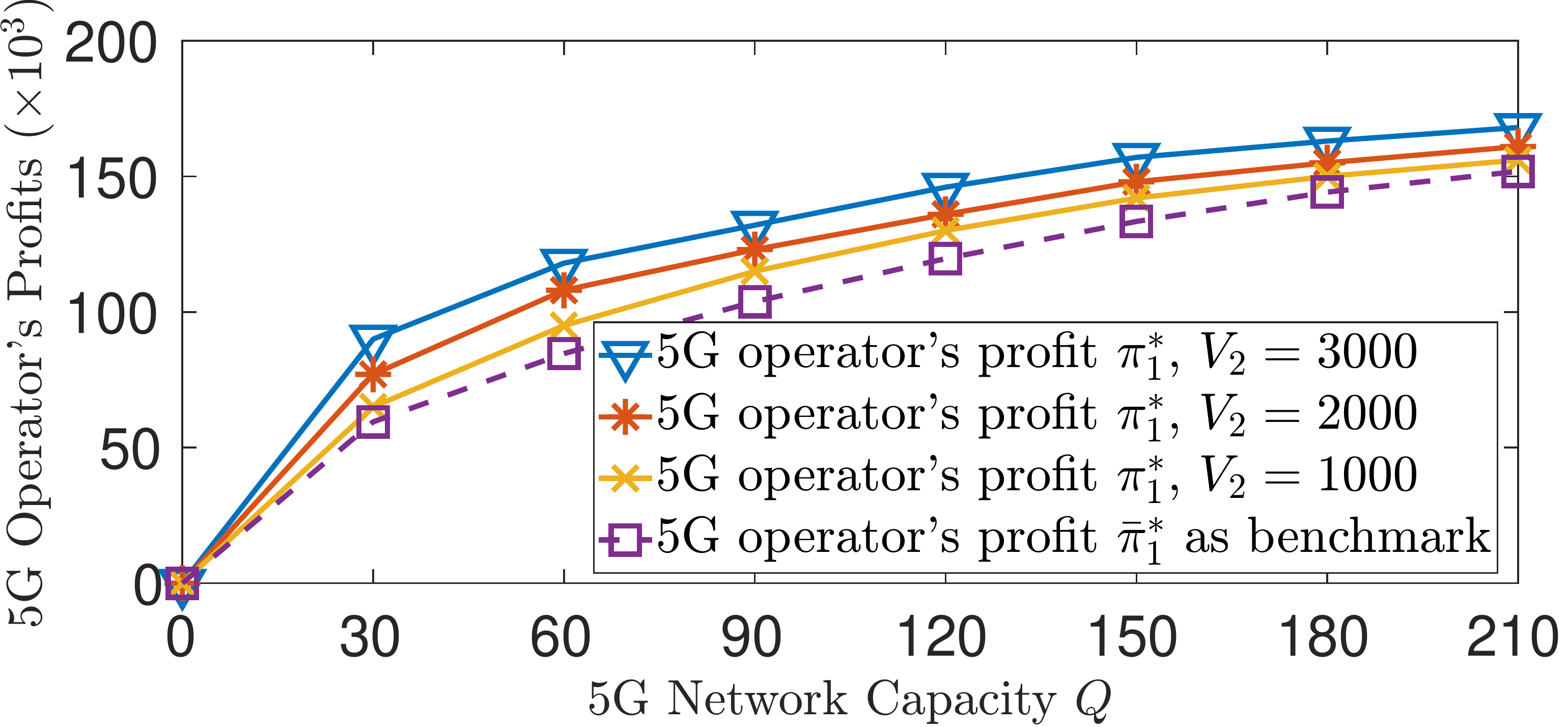}}
\caption{5G operator's profits $\bar{\pi}_1^*$ and $\pi_1^*$ before and after the introduction of the crowdsourced WiFi versus the 5G capacity $Q$ and the WiFi mobile access benefit $V_2$, respectively. Here, 
a user's congestion sensitivity $\theta$ follows truncated distribution $N(0.5, 1)$ in the normalized range [0, 1]. And we fix $N = 10^5$,  $V_1 = 3000$, $\bar{u} = 1000$, $\alpha = 0.5$, and {$c = 50$}.}
\label{figb}
\end{figure}

\subsection{Stage II Analysis on Users' Subscription}\label{section:TMC-7.1}
We begin with Stage II analysis for users' choices. We still expect congestion-insensitive users join the 5G network only, and sensitive ones choose to further add on the crowdsourced WiFi or not. However, the newly considered WiFi congestion may prevent highly sensitive users from adding on, and users may not even consider to pay to add-on WiFi if the congestion there is severe. By following the similar analysis in the proof of Lemma~\ref{lemma:TMC-2}, we have the following.

\begin{lemma}\label{lemma:TMC-4}
After considering the WiFi congestion, at the equilibrium in Stage II we have \eqref{e11} if $x_2^* > 0$ and \eqref{e12} if $x_2^*=0$, as illustrated in Fig. \ref{fig4}. As an extreme case with $\beta\geq N/Q$, no user will consider to add on the crowdsourced WiFi and the user subscription reduces to \eqref{e333}. 
\end{lemma}

Given the equilibrium users' choices in Fig. \ref{fig4} based on cutoff points $F^{-1}(x_1^*)$ and $F^{-1}(x_1^*+x_2^*)$, we are ready to derive $x_1^*$ and $x_2^*$ in the closed-loop to summarize all the users' choices based on their $\theta$'s. 
\begin{lemma}\label{lemma:TMC-5}
After considering the WiFi congestion, at the equilibrium of Stage II, user fractions of the 5G and WiFi networks (i.e., $x_1^*$ and $x_2^*$) are solutions to 
\begin{equation}
u_1(F^{-1}(x_1^*))=\tilde{u}_2(F^{-1}(x_1^*)), \;\; \tilde{u}_2(F^{-1}(x_1^*+x_2^*))=\bar{u}, \label{e34}
\end{equation}
if $x_1^*+x_2^* < 1$ (see Fig. \ref{fig4}(a)), and otherwise 
\begin{equation}
u_1(F^{-1}(x_1^*))=\tilde{u}_2(F^{-1}(x_1^*)), \;\; x_1^*+x_2^*=1, \label{e35}
\end{equation}
with $\tilde{u}_2(F^{-1}(x_1^*+x_2^*))\geq \bar{u}$ (see Fig. \ref{fig4}(b)), where $x_2^*$ is the largest among all the possible solutions. Specifically, we have $x_1^* = \bar{x}_1^*$ and $x_2^* = 0$ if the crowdsourced WiFi is easier to congest than the 5G network (i.e., $\beta \geq N/Q$).
\end{lemma}

The structure of users' subscription equilibrium in Stage II is still similar to Fig. \ref{bg} with $\beta = 0$. But after examining \eqref{e34} and \eqref{e35}, the WiFi congestion require both 5G and WiFi prices to be smaller to have full subscription  (i.e., $x_1^* + x_2^* = 1$) and positive WiFi subscription (i.e., $x_2^* > 0$, $x_1^* + x_2^* < 1$), and $x_2^*$ reduces due the WiFi congestion.

\begin{figure}
\centerline{\includegraphics[scale=0.28]{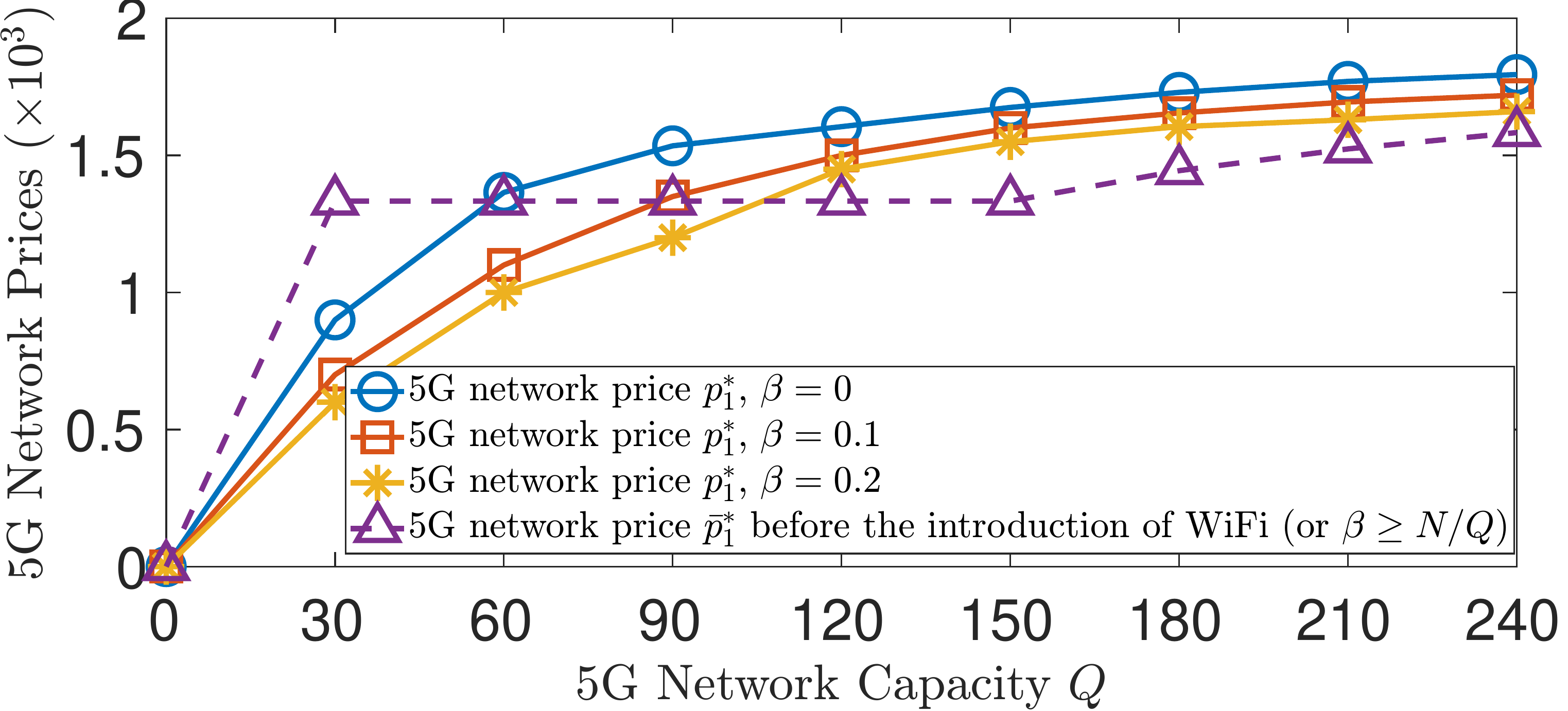}}
\caption{5G network's annual prices $\bar{p}_1^*$ in \eqref{p_5} and $p_1^*$ before and after the introduction of the crowdsourced WiFi, versus the 5G capacity $Q$. Here we fix $N = 10^5$, {$V = 3000$}, $\bar{u} = 1000$, {$c = 50$}, and $\alpha=0.5$ with $3N/(V-\bar{u}) = 150$. Note that the benchmark $\bar{p}_1^*$ before the introduction of the crowdsourced WiFi is equivalent to the crowdsourced WiFi case with $\beta\geq N/Q$.}
\label{fig100}

\end{figure}

\subsection{Stage I Analysis on 5G Price and Profit}\label{section:TMC-7.2}
Now we turn to Stage I to analyze the 5G operator's price and profit at the equilibrium. By using a similar analysis as in Lemma~\ref{lemma:TMC-3}, we have the similar result.

\begin{proposition}\label{prop:TMC-10}
After considering the WiFi congestion, at the equilibrium of the whole dynamic game, the 5G operator obtains at least the same profit after the introduction of the crowdsourced WiFi, i.e., $\pi_1^* \geq \bar{\pi}_1^*$ with $\bar{\pi}_1^*$ in \eqref{pi_5}. 
\end{proposition}
\begin{proof}
See Appendix G.
\end{proof}

In the extreme case of severe WiFi congestion $\beta \geq N/Q$, according to Lemma~\ref{lemma:TMC-4}, $x_2^* = 0$ and the 5G network gains exactly the same profit after the introduction of the crowdsourced WiFi. While $\beta$ is small, some congestion-sensitive users will add-on WiFi to help offload 5G traffic, reducing the 5G congestion and enabling the 5G operator to charge a higher price to gain more profit. Note that as $\beta$ becomes trivial, we have the same results as in Propositions~\ref{prop:TMC-5}-\ref{prop:TMC-7}.   

Due to more involved analysis, next  we use numerical results to examine the impact of WiFi congestion.

Figure \ref{fig100} shows the 5G network's equilibrium prices $\bar{p}_1^*$ in \eqref{p_5} before introducing WiFi and $p_1^*$ after introducing WiFi, as functions of the 5G capacity $Q$ and WiFi congestion factor $\beta$, respectively. Here, $Q$ ranges from 0 to 240, and $\beta< N/Q$ holds for the three pricing curves after introducing WiFi. Then we observe that  $p_1^*$ is smaller than $\bar{p}_1^*$ for small $Q$ to motivate users' WiFi offloading, and it is greater than $\bar{p}_1^*$ for non-small $Q$.  This is similar to Propositions~\ref{prop:TMC-5}, \ref{prop:TMC-8}, \ref{prop:TMC-9} and Fig. \ref{fig10011} without WiFi congestion. 

Figure \ref{fig100} also shows that all 5G prices decrease with WiFi congestion factor $\beta$ as long as $\beta<N/Q$. As $\beta$ increases, less users will choose to add on WiFi to offload and the 5G congestion is less mitigated. Thus, the 5G price has to be reduced to attract users.

Figure \ref{fig11} shows the 5G operator's equilibrium profits $\bar{\pi}_1^*$ in \eqref{pi_5} before introducing WiFi and $\pi_1^*$ after introducing WiFi, as functions of $Q$ and the crowdsourced WiFi's congestion factor $\beta$. All 5G profits are increasing in $Q$ due to the increased 5G pricing in Fig. \ref{fig100}.

Figure \ref{fig11} also shows that the 5G profit improves from the traditional 5G-only case due to the crowdsoruced WiFi, which is consistent with Proposition~\ref{prop:TMC-10}. As $\beta$ increases, this profit gain for the 5G operator becomes less obvious. As an extreme case of $\beta\geq N/Q$, there is no difference from the benchmark case of 5G only.

 \begin{figure}
\centerline{\includegraphics[scale=0.28]{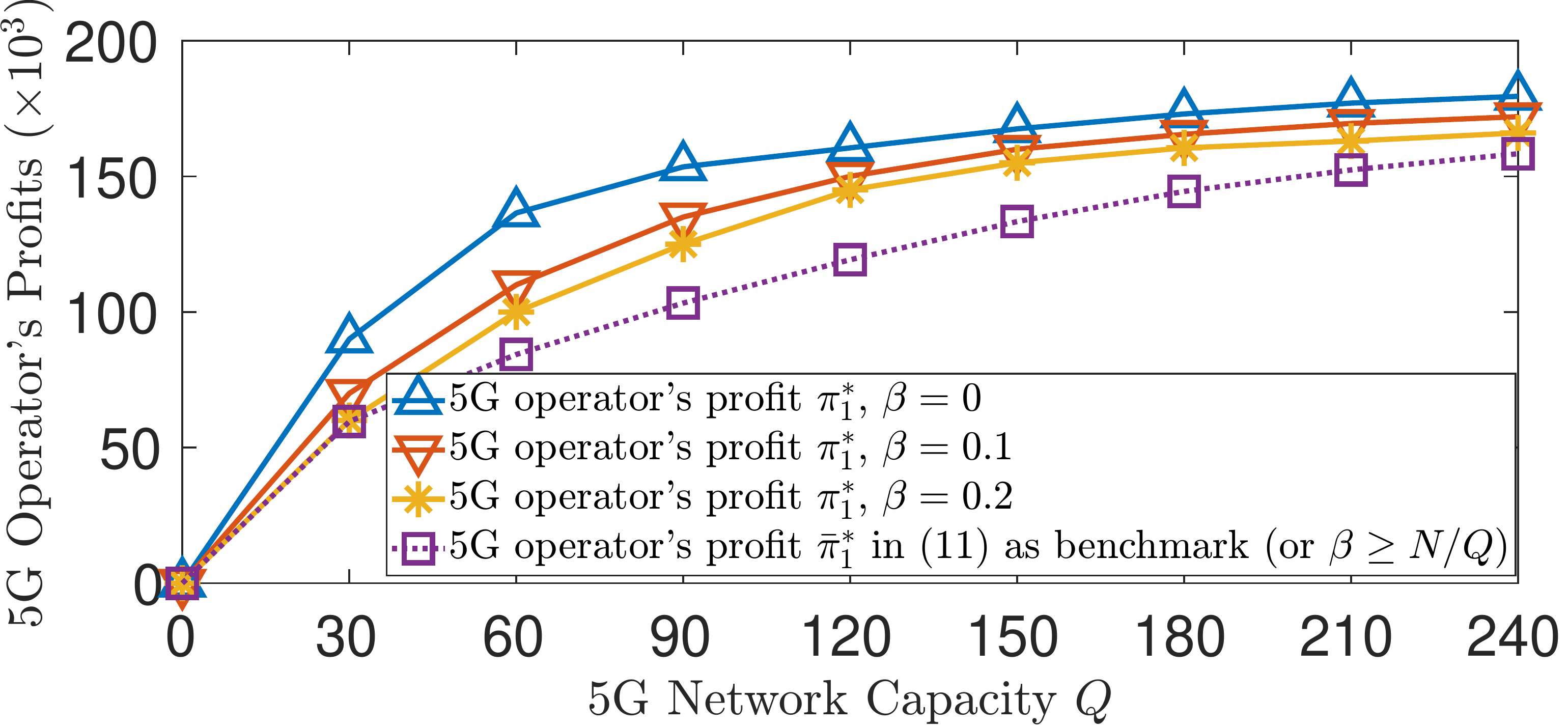}}
\caption{5G operator's profits $\bar{\pi}_1^*$ in \eqref{pi_5} and $\pi_1^*$ before and after the introduction of the crowdsourced WiFi. Here we vary the 5G capacity $Q$ and the WiFi congestion factor $\beta$, and fix $N = 10^5$, { $V = 3000$}, $\bar{u} = 1000$, {$c = 50$}, and $\alpha = 0.5$ with $3N/(V-\bar{u}) = 150$. Note that the benchmark $\bar{\pi}_1^*$ before the introduction of the crowdsourced WiFi is equivalent to the crowdsourced WiFi case with $\beta\geq N/Q$.}
\label{fig11}
\end{figure}

\section{Conclusion}\label{section:TMC-8}

In this paper, we investigate how the 5G and crowdsourced WiFi networks of positive and negative externalities, respectively, co-exist from an economic perspective. We propose a dynamic game theoretic model to analyze the interaction among the 5G operator, the crowdsourced WiFi operator, and users: in the first stage the 5G and WiFi operators decide their service prices for maximizing their own profits according to the network capacity and deployment cost, and in the second stage users of different congestion sensitivities decide to subscribe to 5G, 5G+WiFi or neither. Our user choice model with WiFi's complementarity for 5G allows users to choose both services, and departs from the traditional economics literature where a user chooses one over another alternative. By backward induction, we analytically prove at the equilibrium that the 5G operator gains more profit. To motivate users to add on the WiFi service for WiFi offloading, we prove that the 5G operator with small capacity will purposely lower his price to his users. We also prove that  the 5G operator with non-small capacity will charge users more after the introduction of crowdsourced WiFi, and all the users' payoffs can decrease. Furthermore, our main results on user subscription, and 5G price and profit still hold when considering WiFi congestion.

{In the future, there are some directions to extend this work. For example,  it may be interesting to extend the 5G operator's  flat-rate pricing to usage-based pricing. The usage-based pricing scheme better controls users' data consumption in 5G, yet the 5G operator still needs WiFi offloading to keep its service quality.   }



%




\vspace{-30pt}

\begin{IEEEbiography}[{\includegraphics[width=1in,height=1.25in,clip,keepaspectratio]{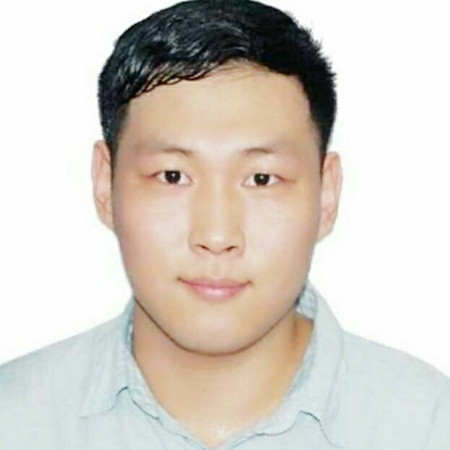}}]{Shugang Hao}
 is currently working toward the PhD degree in the Pillar of Engineering and System Design at Singapore University of Technology and Design. He received B.Eng. degree in Information Engineering from South China University of Technology in 2017. 
His research interests are network economics and mechanism design. 
His works appeared in IEEE JSAC and ACM MobiHoc symposium. 
\end{IEEEbiography}


\begin{IEEEbiography}[{\includegraphics[width=1in,height=1.25in,clip,keepaspectratio]{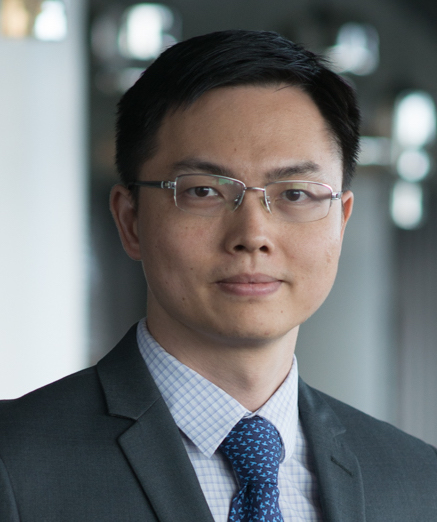}}]{Lingjie Duan}
(S'09-M'12-SM'17) received the
Ph.D. degree from The Chinese University of
Hong Kong in 2012. He is an Associate Professor
of Engineering Systems and Design with the
Singapore University of Technology and Design
(SUTD). In 2011, he was a Visiting Scholar at University
of California at Berkeley, Berkeley, CA,
USA. His research interests include network economics
and game theory, cognitive communications,
and cooperative networking. He is an
Editor of IEEE Transactions on Wireless Communications. He also served as a Guest
Editor of the IEEE Journal on Selected Areas in Communications. He received the SUTD Excellence
in Research Award in 2016 and the 10th IEEE ComSoc Asia-Pacific
Outstanding Young Researcher Award in 2015.
\end{IEEEbiography}

\vfill

\newpage


\appendices

\section{Proof of Corollary~\ref{coro:TMC-1}}\label{section: TMC-Appendix-A}
In Corollary~\ref{coro:TMC-1}, the requirement for $x_2^* > 0$ can be obtained from Proposition~\ref{prop:TMC-3} and Fig. \ref{bg}, and we skip the details here. 
To show the right-handed side (denoted as $g_1(\alpha)$) of \eqref{c_eq}  increases with $\alpha$, we first check its partial derivative with respect to $\alpha$ as follows:
\begin{equation*}
   g_1'(\alpha) = \frac{18\alpha^2\frac{V_1-\bar{u}}{3N/Q} + 6\alpha\sqrt{\frac{V_1-\bar{u}}{3N/Q}} + 4\sqrt{1-3\alpha \sqrt{\frac{V_1-\bar{u}}{3N/Q}}}-4}{54\alpha^2\sqrt{1-3\alpha\sqrt{\frac{V_1-\bar{u}}{3N/Q}}}\sqrt{\frac{3Q}{N(V_1-\bar{u})}}}. 
\end{equation*}
To prove $g_1'(\alpha) \geq 0$, it is sufficient to prove the numerator of $g_1'(\alpha)$ is non-negative. 
By denoting $y := \alpha \sqrt{\frac{V_1-\bar{u}}{3N/Q}}\in[0, \frac{1}{3}]$ (to validate $1-3\alpha \sqrt{\frac{V_1-\bar{u}}{3N/Q}} \geq 0$ in the numerator), we rewrite the numerator of $g_1'(\alpha)$ as $g_2(y)$:
\begin{align*}
    g_2(y) = 18y^2 + 6y + 4(\sqrt{1-3y} - 1).
\end{align*}
To prove $g_2(y) \geq 0$ for $y \in [0, \frac{1}{3}]$, we further check the first and second derivatives of $g_2(y)$ over $y$ as:
\begin{equation*}
     g_2'(y) = 36y + 6 - \frac{6}{\sqrt{1-3y}}, \; g_2''(y) = 36  - \frac{9}{(1-3y)^\frac{3}{2}}.
\end{equation*}
We then have $g_2''(y) \geq 0$ for $y \leq \frac{1 - (\frac{1}{16})^\frac{1}{3}}{3}$ and have $g_2''(y) < 0$ for $y \in (\frac{1 - (\frac{1}{16})^\frac{1}{3}}{3}, \frac{1}{3}]$, which means $g_2'(y)$ first increases with $y \leq \frac{1 - (\frac{1}{16})^\frac{1}{3}}{3}$  and then decreases with $y \in (\frac{1 - (\frac{1}{16})^\frac{1}{3}}{3}, \frac{1}{3}]$. Since we have
\begin{equation}
     g_2'(0) = 0, \; g_2'(\frac{1}{3}) = -\infty, \nonumber
\end{equation}
there exists the unique solution $y_0 \in (0, \frac{1}{3})$ to continuous equation $g_2'(y) = 0$, and we have $g_2'(y) \geq 0$ for $y \leq y_0$ and $g_2'(y) < 0$ for $y \in (y_0, \frac{1}{3}]$. Therefore, $g_2(y)$ first increases with $y \leq y_0$ and then decreases with $y \in (y_0, \frac{1}{3}]$. Since $g_2(0) = g_2(\frac{1}{3}) = 0$, we have $g_2(y) \geq 0$ for any $y \in [0, 1/3]$. Accordingly, the numerator of $g_1'(\alpha)$ and $g_1'(\alpha)$ are both non-negative, and  $g_1(\alpha)$ or the right-hand-side of \eqref{c_eq} thus increases with $\alpha$.

\section{Proof of Proposition~\ref{prop:TMC-4}}\label{section: TMC-Appendix-B}
According to Corollary~\ref{coro:TMC-1}, if $Q < 3N/(V_1-\bar{u})$ and \eqref{c_eq} holds, we have $x_2^* > 0$ for $p_1 = \bar{p}_1^*$. Therefore,  it is sufficient to prove $\pi_1^* > \bar{\pi}_1^*$ under $x_2^* > 0$ and $Q < 3N/(V_1-\bar{u})$ in the special case of the 5G operation with $p_1 = \bar{p}_1^*$.  We then compare $\pi_1^*$ with $\bar{\pi}_1^*$, depending on whether $x_1^* + x_2^*=1$ or not.
\begin{itemize}
    \item Case of $x_1^* + x_2^* = 1$: We have $x_1^* + x_2^* > \bar{x}_1^*$ due to $\bar{x}_1^* < 1$ under $Q < 3N/(V_1-\bar{u})$ in \eqref{x_5}. Therefore, the 5G operator's subscriber fraction $x_1^*+x_2^*$ increases at the same price $p_1=\bar{p}_1^*$, leading to a strictly greater profit. 
 \item  Case of $x_1^* + x_2^* < 1$: According to \eqref{e11} in Lemma~\ref{lemma:TMC-2}, users with $\theta \leq x_1^*$ join 5G network only and those with $\theta \in (x_1^*, x_1^* + x_2^*]$ join 5G+WiFi. Since $x_1^* + x_2^* < 1$ and $p_1=\bar{p}_1^*$, we rewrite the two equations in \eqref{e7} as:
\begin{equation}
    p_2 = \alpha x_2^* \frac{N}{Q} \big(x_1^* + x_2^* (1 - \alpha x_2^*)\big)x_1^*, \label{a25}
    \end{equation}
    \begin{equation}
    V_1 \!-\! \frac{N}{Q}(1 - \alpha x_2^*) \big(x_1^* + x_2^*(1 - \alpha x_2^*)\big) (x_1^* + x_2^*) \!  - \! \bar{p}_1^* \! -\! p_2\! =\! \bar{u}. \label{a26}
\end{equation}
We substitute \eqref{a25} into \eqref{a26}, and simplify \eqref{a26} as:
\begin{equation}
    V_1 -\frac{N}{Q} \big(x_1^*+x_2^*(1 - \alpha x_2^*)\big)^2 - \bar{p}_1^* = \bar{u}. \label{a23}
\end{equation}
Recall that before introducing the crowdsourced WiFi, we have the cutoff user with $\theta = \bar{x}_1^*$ to obtain the fixed reservation payoff and is indifferent for  joining the 5G network or not:
\begin{equation}
    \bar{u}_1(\theta = \bar{x}_1^*) = V_1 - \frac{N}{Q} (\bar{x}_1^*)^2 - \bar{p}_1^* = \bar{u}. \label{a24}
\end{equation}
By comparing \eqref{a23} and \eqref{a24}, we conclude 
\begin{align*}
    x_1^* + x_2^*(1 - \alpha x_2^*) = \bar{x}_1^*,
\end{align*}
which implies $x_1^* + x_2^* > \bar{x}_1^*$ due to $1 - \alpha x_2^* < 1$. Therefore, the 5G operator's subscriber fraction $x_1^*+x_2^*$ increases at the same price $p_1=\bar{p}_1^*$, leading to a strictly greater profit. 
\end{itemize} 

\section{Proof of Proposition~\ref{prop:TMC-5}}\label{section: TMC-Appendix-C}
If no user adds on the crowdsourced WiFi, i.e., $x_2^* = 0$, the 5G operator faces the same situation as before the introduction of the crowdsourced WiFi. Thus, his price remains unchanged.

If \eqref{by8} holds  and $Q \geq 3N/(V_1-\bar{u})$, we have $x_2^* > 0$ and
 the 5G operator's profit function $\pi_1(p_1, x_2^*)$ in \eqref{a42} first increases with $p_1 \leq V_1 -\bar{u}- \frac{N}{Q}(1-\alpha \bar{x}_2^2)^2$  and we thus have 
 \begin{equation}
     p_1^* \geq V_1 -\bar{u}- \frac{N}{Q}(1-\alpha \bar{x}_2^2)^2 > V_1 - \bar{u}- \frac{N}{Q}\label{pp1}
 \end{equation}
 due to $\bar{x}_2 > 0$.
 
 
 If $Q \geq \frac{3N}{V_1-\bar{u}}$, from \eqref{p_5} we have $\bar{p}_1^* = V_1 -\bar{u} - N/Q$ before the introduction of the crowdsourced WiFi. Together with \eqref{pp1}, we have $p_1^*> \bar{p}_1^*$. 
 
\section{Proof of Proposition~\ref{prop:TMC-7}}\label{section: TMC-Appendix-D}
 According to Proposition~\ref{prop:TMC-5}, we have the 5G operator's price $p_1^* > \bar{p}_1^*$ and $x_2^* > 0 $  if $Q \geq \frac{3N}{V_1-\bar{u}}$ and \eqref{by8} holds. We then prove $\max\{u_1^*(\theta), u_2^*(\theta)\}<\bar{u}_1^*(\theta)$ according to $x_1^* + x_2^* = 1$ or not.
\begin{itemize}
    \item Case of  $x_1^*+x_2^* = 1$: with full participation, to maximize $\pi_1$ in \eqref{a42}, we have $x_2^* = \bar{x}_2$ and
\begin{equation*}
    p_1^*(x_2^*) = V_1 - \bar{u} -\frac{N}{Q} (1 - \alpha (x_2^*)^2)^2.
\end{equation*}
By substituting this price into $u_1(\theta)$ in \eqref{e3}, we have users' payoff with $\theta \leq x_1^*$ after the introduction of  the crowdsourced WiFi as:
\begin{equation}
    u_1^*(\theta) = \frac{N}{Q}(1-\alpha(x_2^*)^2)(1-\alpha(x_2^*)^2 - \theta) + \bar{u}. \label{n1}
\end{equation}
We then substitute $\bar{p}_1^*$ in \eqref{p_5} and  $\bar{x}_1^*$ in \eqref{x_5} under $Q \geq \frac{3N}{V_1-\bar{u}}$ into benchmark payoff $\bar{u}_1(\theta)$ in \eqref{e2} before the introduction of the crwodsourced WiFi:
\begin{equation}
    \bar{u}_1^*(\theta) = \frac{N}{Q}(1-\theta) + \bar{u}.\label{n2}
\end{equation}
By comparing \eqref{n1} and \eqref{n2}, we have $u_1^*(\theta) < \bar{u}_1^*(\theta)$ for any $\theta \leq x_1^*$ due to $x_2^* > 0$. 
Similarly, users with congestion sensitivity $\theta \in (x_1^*, 1]$ choose 5G+WiFi service and we rewrite payoff $u_2(\theta)$ in \eqref{e4} as:
\begin{equation}
    u_2^*(\theta) = \frac{N}{Q}(1-\alpha(x_2^*)^2)(1-\theta - \alpha x_2^* (1-\theta)) + \bar{u}.\label{n3}
\end{equation}
Comparing \eqref{n2} and \eqref{n3}, we have $u_2^*(\theta) < \bar{u}_1^*(\theta)$ for any $\theta \in (x_1^*, 1]$ due to $x_2^* > 0$. We then conclude $\max\{u_1^*(\theta), u_2^*(\theta)\}<\bar{u}_1^*(\theta)$ for any $\theta$.

\item Case of  $x_1^*+x_2^* < 1$: It is sufficient to prove $u_1^*(\theta) < \bar{u}_1^*(\theta)$ for $\theta \leq x_1^*$ and $u_2^*(\theta) < \bar{u}_1^*(\theta)$ for $\theta \in (x_1^*, x_1^*+x_2^*]$ under $p_1^* > \bar{p}_1^*$ and $x_2^* > 0$ due to $\bar{u}_1^*(\theta) \geq \bar{u} > \max\{u_1^*(\theta), u_2^*(\theta)\}$ for $\theta \in (x_1^* + x_2^*, 1]$.
For $\theta \leq x_1^*$, we have users' payoff $u_1^*(\theta)$ in \eqref{e3} after the introduction of the crowdsourced WiFi as follows:
\begin{equation}
    u_1^*(\theta) 
    \!=\! \frac{N}{Q}(x_1^* + x_2^*(1-\alpha x_2^*))(x_1^* + x_2^*(1-\alpha x_2^*)-\theta) \!+\!\bar{u} \label{re2}
\end{equation}
 due to 
 \begin{equation*}
    x_1^* + x_2^*(1 - \alpha x_2^*) =  \sqrt{\frac{V_1-\bar{u}-p_1^*}{N/Q}}.
\end{equation*}
 By comparing \eqref{n2} and \eqref{re2}, we obtain $u_1^*(\theta) < \bar{u}_1^*(\theta)$ for $\theta \leq x_1^*$ due to $x_1^* + x_2^* < 1$.

In the following we will prove $u_1^*(\theta) < \bar{u}_1^*(\theta)$ for $\theta \in (x_1^*, x_1^* + x_2^*]$. For $\theta \in (x_1^*, x_1^* + x_2^*]$, we have users' payoff $u_2^*(\theta)$ in \eqref{e4} after the introduction of the crowdsourced WiFi as follows:
\begin{align}
    &u_2^*(\theta) \nonumber \\
    &=\!  V_1 \!-\! p_1^*\! -\! \frac{N}{Q} \sqrt{\frac{V_1-\bar{u}-p_1^*}{N/Q}} \left((1-\alpha x_2^*)\theta + \alpha x_2^* x_1^*\right) \label{re6} \\
    &= \!\frac{N}{Q}(x_1^* + x_2^*(1-\alpha x_2^*))(x_1^* + x_2^* -\theta)(1-\alpha x_2^*)\! +\! \bar{u}\label{re7},
\end{align}
where \eqref{re6} holds due to 
\begin{equation*}
    p_2^*  = \alpha x_2^* \frac{N}{Q} \left(x_1^* + x_2^*(1-\alpha x_2^*)\right)x_1^*
\end{equation*}
and \eqref{re7} due to 
\begin{equation*}
    x_1^* + x_2^*(1 - \alpha x_2^*) =  \sqrt{\frac{V_1-\bar{u}-p_1^*}{N/Q}}.
\end{equation*}
By comparing \eqref{n2} and \eqref{re7}, we obtain $u_2^*(\theta) < \bar{u}_1^*(\theta)$ for $\theta \in (x_1^*, x_1^*+x_2^*]$ due to $x_1^* + x_2^* < 1$. We then conclude $\max\{u_1^*(\theta), u_2^*(\theta)\}<\bar{u}_1^*(\theta)$ for any $\theta$.
 
\end{itemize}

\section{Proof of Proposition~\ref{prop:TMC-8}}\label{section: TMC-Appendix-E}
\eqref{alpha} is equal to 
\begin{align}
    &4\alpha \bigg(\sqrt{\frac{1-\sqrt{\frac{V_1-\bar{u}}{3N/Q}}}{\alpha}}\bigg)^3 - 3\alpha\bigg(\sqrt{\frac{1-\sqrt{\frac{V_1-\bar{u}}{3N/Q}}}{\alpha}}\bigg)^2 \nonumber \\
    &- 2\sqrt{\frac{1-\sqrt{\frac{V_1-\bar{u}}{3N/Q}}}{\alpha}} + 1 > 0. \label{by2}
\end{align}
Denote $h(x_2) := 4\alpha x_2^3 - 3 \alpha x_2^2 - 2 x_2 + 1$, it is easy to prove $h(x_2) \geq 0$ for $x_2 \in [0, \hat{x}_2]$ and $h(x_2) < 0$ for $x_2 \in (\hat{x}_2, 1]$, where $\hat{x}_2 \in [0, 1]$ as the unique solution to \eqref{hatx}.
Together with \eqref{by2}, we have $ \hat{x}_2 > \sqrt{\frac{1-\sqrt{\frac{V_1-\bar{u}}{3N/Q}}}{\alpha}}$.
Since in \eqref{a42} $\bar{x}_2 \geq \hat{x}_2$, we further have $\bar{x}_2 > \sqrt{\frac{1-\sqrt{\frac{V_1-\bar{u}}{3N/Q}}}{\alpha}}$, which is equal to
\begin{equation}
    V_1 - \bar{u}- \frac{N}{Q}(1-\alpha \bar{x}_2^2)^2 > \frac{2}{3}(V_1-\bar{u}). \label{by3}
\end{equation}

With $c$ in \eqref{by8} and Proposition~\ref{prop:TMC-3}, we have $x_2^* > 0$. Recall that the profit objective of the 5G operator $\pi_1(p_1, x_2^*)$ as a function of his price $p_1$ and users' subscription to 5G+WiFi $x_2^*$ in \eqref{a42},
we then have the 5G operator's profit function $\pi_1(p_1, x_2^*)$ in \eqref{a42} first increases with $p_1 \leq V_1 - \bar{u} - \frac{N}{Q}(1-\alpha (\bar{x}_2)^2)^2$ and thus
\begin{equation}
    p_1^* \geq V_1 - \bar{u} - \frac{N}{Q}(1-\alpha \bar{x}_2^2)^2. \label{by1}
\end{equation}
With \eqref{by3} and \eqref{by1}, we have $p_1^* > \frac{2}{3}(V_1-\bar{u})$.
Since we have $\bar{p}_1^* = \frac{2}{3}(V_1-\bar{u})$ from Proposition~\ref{prop:TMC-1}, we conclude $p_1^* > \bar{p}_1^*$ under conditions in Proposition~\ref{prop:TMC-8}.

\section{Proof of Proposition~\ref{prop:TMC-9}}\label{section: TMC-Appendix-F}
 If $Q < N/(V_1-\bar{u})$, $\alpha < 1 - \sqrt{\frac{V_1-\bar{u}}{N/Q}}$ and \eqref{by5} holds,  we have the profit objective of the 5G operator as a function of his price $p_1$ and users' subscription to 5G+WiFi $x_2^*$ is
\begin{align*}
\pi_1(p_1, x_2^*) &= 
  N p_1 \!\big( \!\sqrt{\frac{V_1-\bar{u}-p_1}{N/Q}} + \alpha (x_2^*)^2 \!\big), \nonumber \\
  &0 < \!p_1  < V_1 - \bar{u} -\frac{N}{Q}(x_2^*)^2(1- \alpha x_2^*)^2, 
 \label{by4}  
\end{align*}
and the crowdsourced WiFi's profit objective $\pi_2$ is
\begin{equation*}
    \pi_2(p_1, p_2) = 
  N (p_2 - c) x_2^*, \;\; p_2 \geq c, \; 
 \label{by6}  
\end{equation*}
where we have
\begin{align*}
    x_2^* = 
    \begin{cases}
    \bigg\{x_2: \alpha x_2 \frac{N}{Q} \sqrt{\frac{V_1-\bar{u}-p_1}{N/Q}}\big(\sqrt{\frac{V_1-\bar{u}-p_1}{N/Q}}\!\!\!\!&- x_2 +\alpha x_2^2\big) \!-\! p_2
      \\    
      = 0, x_2 > \frac{1 - \sqrt{1 - 3\alpha \sqrt{\frac{V_1-\bar{u}-p_1}{N/Q}}}}{3\alpha}.\bigg\}, &\!\!\!\text{if} \; p_1 < \frac{2}{3}(V_1-\bar{u}),\\
    0,  &\!\!\!\text{if} \; p_1 \geq \frac{2}{3}(V_1-\bar{u}).
    \end{cases}
\end{align*}
according to \eqref{by5} and Proposition~\ref{prop:TMC-3}. Note that if $p_1 \geq \frac{2}{3}(V_1-\bar{u})$, we have $\pi_1^* = \bar{\pi}_1^*$ with $p_1^* = \bar{p}_1^* = \frac{2}{3}(V_1-\bar{u})$ due to $x_2^* = 0$. If $p_1 < \frac{2}{3}(V_1-\bar{u})$, we have
\begin{align}
    &\pi_1(p_1, x_2^*) = 
  N p_1 \bigg( \sqrt{\frac{V_1-\bar{u}-p_1}{N/Q}} + \alpha (x_2^*)^2 \bigg) \nonumber \\
  &> N p_1 \bigg( \sqrt{\frac{V_1-\bar{u}-p_1}{N/Q}} + \frac{\big(1 - \sqrt{1 - 3\alpha \sqrt{\frac{V_1-\bar{u}-p_1}{N/Q}}}\big)^2}{9\alpha} \bigg) \label{by66}
\end{align}
due to $x_2^* > \frac{1 - \sqrt{1 - 3\alpha \sqrt{\frac{V_1-\bar{u}-p_1}{N/Q}}}}{3\alpha}$. Denote right-handed side of \eqref{by66} as
\begin{equation*}
    l(p_1) \!:=\! N p_1 \! \bigg( \! \sqrt{\frac{V_1-\bar{u}-p_1}{N/Q}} + \frac{\big(1 - \sqrt{1 - 3\alpha \sqrt{\frac{V_1-\bar{u}-p_1}{N/Q}}}\big)^2}{9\alpha} \! \bigg), 
\end{equation*}
By taking first derivative of $l(p_1)$ at $p_1 = \frac{2}{3}(V_1-\bar{u})$, we have
\begin{align*}
    l'(\frac{2}{3}(V_1-\bar{u})) = &-\frac{N}{9\alpha}
    \frac{\bigg(1 - \sqrt{1 - 3\alpha \sqrt{\frac{V_1-\bar{u}}{3N/Q}}}\bigg)^2}{\sqrt{1-3\alpha\sqrt{\frac{V_1-\bar{u}}{3N/Q}}}} < 0,
\end{align*}
which implies there exists $p_1' < \frac{2}{3}(V_1-\bar{u})$ leading to $l(p_1') > l(\frac{2}{3}(V_1-\bar{u}))$ since $l(p_1)$ is continuous. Since we have 
\begin{align*}
    &l(\frac{2}{3}(V_1-\bar{u})) \\
    &= N \frac{2}{3}(V_1-\bar{u}) \bigg( \sqrt{\frac{V_1-\bar{u}}{3N/Q}} + \frac{\big(1 - \sqrt{1 - 3\alpha \sqrt{\frac{V_1-\bar{u}}{3N/Q}}}\big)^2}{9\alpha} \bigg) \\
    &> N \frac{2}{3}(V_1-\bar{u}) \sqrt{\frac{V_1-\bar{u}}{3N/Q}} = \bar{\pi}_1^*.
\end{align*}
Together with \eqref{by66} and $l(p_1') > l(\frac{2}{3}(V_1-\bar{u}))$, we have
\begin{equation*}
    \pi_1(p_1', x_2^*(p_1')) > \bar{\pi}_1^*,
\end{equation*}
which implies $p_1^* < \frac{2}{3}(V_1-\bar{u})$. 
Since $\bar{p}_1^* = \frac{2}{3} (V_1-\bar{u})$ from Proposition~\ref{prop:TMC-1}, we conclude $p_1^* < \bar{p}_1^*$ under conditions in Proposition~\ref{prop:TMC-9}.

\section{Proof of Proposition~\ref{prop:TMC-10}}\label{section: TMC-Appendix-G}
It is sufficient to prove $\pi_1^* \geq \bar{\pi}_1^*$ in the special case of the 5G operation with  $p_1 = \bar{p}_1^*$. We divide our analysis into the following three cases, depending on $(x_1^*, x_2^*)$. 

\textit{Case 1.} If $x_2^* = 0$, telling that no user adds on the WiFi option at the equilibrium, the 5G operator faces the same situation as before the introduction of the crowdsourced WiFi. Thus, his profit remains unchanged.

\textit{Case 2.} If $x_2^* > 0$, $x_1^* + x_2^*(1  - \alpha x_2^*) \leq \bar{x}_1^*$, we have $x_1^* < \bar{x}_1^*$ due to $x_2^*(1  - \alpha x_2^*) > 0$ in the second condition above. According to Lemma~\ref{lemma:TMC-5}, users with $\theta > x_1^*$ have $\tilde{u}_2(\theta) > u_1(\theta)$. Therefore, we have $\tilde{u}_2(\bar{x}_1^*) > u_1(\bar{x}_1^*)$. Since
\begin{equation}
   u_1(\bar{x}_1^*)\! = \!V \! -\!\frac{N}{Q}(x_1^* + x_2^*( 1 - \alpha x_2^* ))\bar{x}_1^* - \bar{p}_1^* \!
   \geq \!V \!- \!\frac{N}{Q}(\bar{x}_1^*)^2 - \bar{p}_1^*, \label{a21E}
\end{equation}
where the inequality holds due to $x_1^* + x_2^*(1  - \alpha x_2^*) \leq \bar{x}_1^*$. According to Lemma~\ref{lemma:TMC-1}, before the introduction of the crowdsourced WiFi, we have users with $\theta = \bar{x}_1^*$ obtain the fixed reservation payoff from joining the 5G network, that is,  
\begin{equation}
 \bar{u}_1(\bar{x}_1^*) = V - \frac{N}{Q}(\bar{x}_1^*)^2 - \bar{p}_1^* = \bar{u}.   \label{a22E}
\end{equation}
Combining \eqref{a21E} and \eqref{a22E}, we have $u_1(\bar{x}_1^*) \geq \bar{u}_1(\bar{x}_1^*) = \bar{u}$. Since $\tilde{u}_2(\bar{x}_1^*) > u_1(\bar{x}_1^*)$, we have $\tilde{u}_2(\bar{x}_1^*) > \bar{u}$.
If $\bar{x}_1^* = 1$, we have $\tilde{u}_2(1) > \bar{u}$. Since $\tilde{u}_2(\theta)$ in \eqref{new_u2} is decreasing in $\theta$, we have $\tilde{u}_2(\theta) > \bar{u}$ for $\theta > x_1^*$ and thus have $x_1^* + x_2^* = \bar{x}_1^* = 1$. If $\bar{x}_1^* < 1$ and $x_1^* + x_2^* = 1$, we have $x_1^* + x_2^* > \bar{x}_1^*$. If $\bar{x}_1^* < 1$ and $x_1^* + x_2^* < 1$,  we further notice that $\tilde{u}_2(x_1^* + x_2^*) = \bar{u}$ with $x_1^* + x_2^* < 1$ according to \eqref{e34}. Since $\tilde{u}_2(\theta)$ in \eqref{new_u2} is decreasing in $\theta$ and $\tilde{u}_2(\bar{x}_1^*) > \bar{u}$, we have $x_1^* + x_2^* > \bar{x}_1^*$. 
Then we conclude $ x_1^* + x_2^* \geq \bar{x}_1^*$. This tells that the 5G operator's subscriber number non-decreases at the same price $p_1=\bar{p}_1^*$, leading to at least the same profit for the 5G operator. 
 
\textit{Case 3.} If $x_2^* > 0$ and $x_1^* + x_2^*(1  - \alpha x_2^*)  > \bar{x}_1^*$, we have $ x_1^* + x_2^* > \bar{x}_1^*$ due to $1  - \alpha x_2^* < 1$. Therefore, the 5G operator's subscriber number increases at the same price $p_1=\bar{p}_1^*$, leading to a greater profit.

\end{document}